\documentclass[a4paper, amsfonts, amssymb, amsmath, reprint, showkeys, nofootinbib, superscriptaddress]{revtex4-1} 
\usepackage[utf8]{inputenc}
\usepackage[T1]{fontenc}
\usepackage{amssymb}
\usepackage[colorinlistoftodos, color=green!40, prependcaption]{todonotes}
\usepackage{amsthm}
\usepackage{mathtools}
\usepackage{physics}
\usepackage{caption}
\usepackage{subcaption}
\usepackage[left=15mm,right=15mm,top=30mm, columnsep=15pt]{geometry}
\usepackage{adjustbox}
\usepackage{relsize}
\usepackage{placeins}
\usepackage{hyperref}
\usepackage{extarrows}
\usepackage{csquotes}
\DeclarePairedDelimiter\ceil{\lceil}{\rceil}
\DeclarePairedDelimiter\floor{\lfloor}{\rfloor}
\usepackage[]{graphicx}
\usepackage{float}
\usepackage{qcircuit}
\usepackage{amsmath}
\usepackage{braket}
\usepackage{verbatim}
\usepackage{dsfont}
\usepackage{bm}
\usepackage[shortlabels]{enumitem}
\usepackage{multirow}
\usepackage{rotating}
\numberwithin{equation}{section}
\newcommand{\Mod}[1]{\ (\mathrm{mod}\ #1)}

\newtheorem{lemma}{Lemma}
\newtheorem{proposition}{Proposition}
\newtheorem{definition}{Definition}

\begin{document}
\title{Reassessing the computational advantage of quantum-controlled ordering of gates}

\author{Martin J. Renner,}
    \email{martin.renner@univie.ac.at}
    \affiliation{Vienna Center for Quantum Science and Technology (VCQ), Faculty of Physics, University of Vienna, Boltzmanngasse 5, 1090 Vienna, Austria}
    \affiliation{Institute for Quantum Optics and Quantum Information (IQOQI), Austrian Academy of Sciences, Boltzmanngasse 3, 1090 Vienna, Austria}

\author{\v{C}aslav Brukner}
    \affiliation{Vienna Center for Quantum Science and Technology (VCQ), Faculty of Physics, University of Vienna, Boltzmanngasse 5, 1090 Vienna, Austria}
    \affiliation{Institute for Quantum Optics and Quantum Information (IQOQI), Austrian Academy of Sciences, Boltzmanngasse 3, 1090 Vienna, Austria}

\date{February 22, 2021}

\begin{abstract}
Research on indefinite causal structures is a rapidly evolving field that has a potential not only to make a radical revision of the classical understanding of space-time but also to achieve enhanced functionalities of quantum information processing. For example, it is known that indefinite causal structures provide exponential advantage in communication complexity when compared to causal protocols. In quantum computation, such structures can decide whether two unitary gates commute or anticommute with a single call to each gate, which is impossible with conventional (causal) quantum algorithms. A generalization of this effect to $n$ unitary gates, originally introduced in M.~Araújo et al., Phys.~Rev.~Lett.~113, 250402 (2014) and often called Fourier promise problem (FPP), can be solved with the quantum-$n$-switch and a single call to each gate, while the best known causal algorithm so far calls $O(n^2)$ gates. In this work, we show that this advantage is smaller than expected. In fact, we present a causal algorithm that solves the \emph{only} known specific FPP with $O(n \log(n))$ queries and a causal algorithm that solves \emph{every} FPP with $O(n\sqrt{n})$ queries. Besides the interest in such algorithms on their own, our results limit the expected advantage of indefinite causal structures for these problems.
\end{abstract}

\maketitle

\section{Introduction}
One of the most fundamental concepts in science is that of causality: the idea that events occur in a fixed order. It is embedded in the very structure of computation in which operations are performed one after the other. In particular, a quantum circuit is built out of wires, representing the quantum states, and boxes, representing the gates acting on these states in fixed order. However, it was suggested that the interplay between general relativity and quantum theory might require superseding such a paradigm \cite{hardy2005probability, Zych_2019}. Within the last decade, quantum frameworks have been developed that enable the description of indefinite causal structures in which no well-defined global order of events exists~\cite{hardy2005probability, Chiribella_2013, Oreshkov_2012}.

It was observed that the use of indefinite causal structures in information processing can solve certain tasks which cannot be completed by causally ordered quantum circuits \cite{Chiribella_2012} and exponentially reduce the communication cost in communication complexity problems \cite{Guerin_2016}. Furthermore, they can boost the rate of communication through noisy channels \cite{PhysRevLett.120.120502, salek2018quantum, chiribella2018indefinite, Guo2020, goswami2020}, although causal circuits can achieve the same or even better noise reduction \cite{abbott2018communication, Gu_rin_2019, Rubino_2021}. The computational complexity of indefinite causal structures has been studied \cite{Ara_jo_2017, Baumeler_2018} and their experimental accessibility was demonstrated in enhanced quantum photonics experiments \cite{Procopio_2015, Rubino_2017, rubino2017experimental, Goswami_2018, Guerin_experiment, taddei2020experimental}.\\

The most simple example of indefinite causality is based on the quantum switch \cite{Chiribella_2013}. In the quantum switch, two gates act on a target system and the order in which the two gates are applied is controlled by a qubit: if the state of the control qubit is $\ket{0}$, the gate $U_0$ is applied before $U_1$ whereas if the control qubit is in the state $\ket{1}$, the order is reversed. With this quantum-controlled ordering of gates, one can solve certain tasks more efficiently than with any conventional (causal) quantum algorithm. Specifically, one can determine whether two unitary gates commute or anticommute with a single call to each gate, while with any causal quantum algorithm, at least one gate has to be called twice~\cite{Chiribella_2012}.

A generalization of the quantum switch to an arbitrary number of gates is the quantum-$n$-switch. Here, depending on the state of the control system, any permutation of the $n$ gates can be applied on the target system. In order to study the computational power of this quantum-controlled ordering of gates, a promise problem was introduced in Ref.~\cite{1}. This task, which we will call Fourier promise problem (FPP) here, can be solved with the quantum-$n$-switch and a single call to each gate ($n$ queries). At the same time, it was expected that solving the same task with a causal quantum algorithm requires $O(n^2)$ queries. In a recent study, this idea is extended to other promise problems that are easier to realize experimentally~\cite{taddei2020experimental}.\\

In this work, we consider the solutions to the specific and general Fourier promise problems using both the quantum-$n$-switch and causal quantum algorithms. We find that the reduction in the query complexity using the quantum-$n$-switch is smaller than what was assumed so far. More precisely, we present a causal algorithm that solves the {\em only \/} known specific FPP with $O(n\log{n})$ queries and further, a causal algorithm that solves {\em every \/} FPP with $O(n\sqrt{n})$ queries. This reduces the expected advantage of indefinite causal structures in solving this computational task as compared to causal circuits. 

The article is structured as follows: In Section~\ref{FPP}, we give an overview of the Fourier promise problem, the solution with the quantum-$n$-switch and the best causal algorithm that uses $O(n^2)$ queries. In Section~\ref{sectowards}, we derive the property that allows us to find more efficient causal algorithms and give a first example of such an algorithm in Section~\ref{secexample3}. The two main results of this article can be found thereafter. In Section~\ref{seclog}, we present a causal algorithm that solves a specific FPP with $O(n\log{n})$ queries. In Section~\ref{secsqrt}, we give a causal algorithm that solves every FPP with $O(n\sqrt{n})$ queries.


\newcommand{\xyz}{0.15}
\newcommand{\abcd}{1.0}

\begin{figure}[h!]
\centering
\begin{subfigure}{\xyz\textwidth}
\includegraphics[width=\abcd\textwidth]{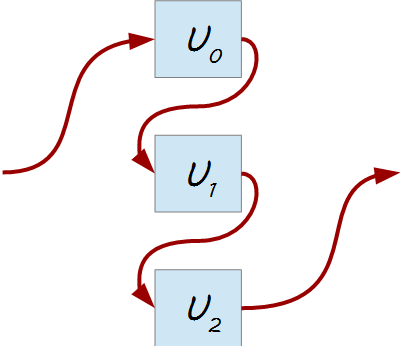}
\subcaption{$\ket{x}_c=\ket{0}$}
\end{subfigure}\hfil
\begin{subfigure}{\xyz\textwidth}
\includegraphics[width=\abcd\textwidth]{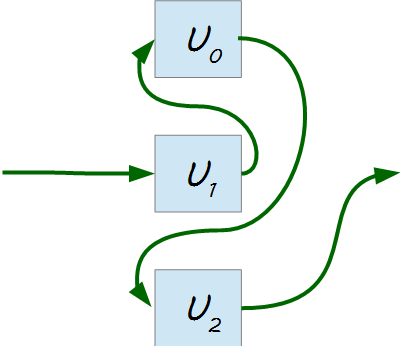}
\subcaption{$\ket{x}_c=\ket{1}$}
\end{subfigure}\hfil
\begin{subfigure}{\xyz\textwidth}
\includegraphics[width=\abcd\textwidth]{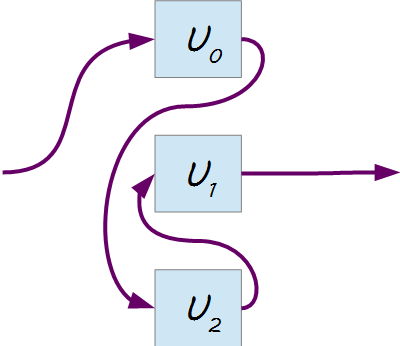}
\subcaption{$\ket{x}_c=\ket{2}$}
\end{subfigure}

\medskip
\begin{subfigure}{\xyz\textwidth}
\includegraphics[width=\abcd\textwidth]{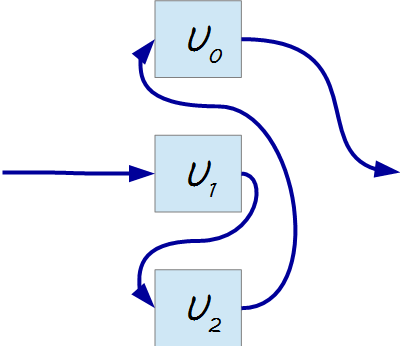}
\subcaption{$\ket{x}_c=\ket{3}$}
\end{subfigure}\hfil
\begin{subfigure}{\xyz\textwidth}
\includegraphics[width=\abcd\textwidth]{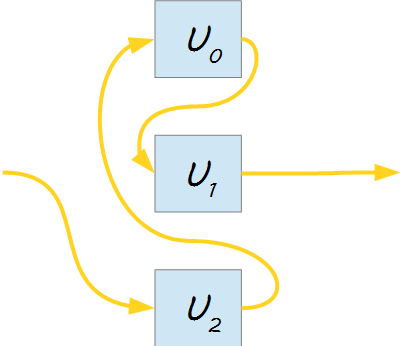}
\subcaption{$\ket{x}_c=\ket{4}$}
\end{subfigure}\hfil
\begin{subfigure}{\xyz\textwidth}
\includegraphics[width=\abcd\textwidth]{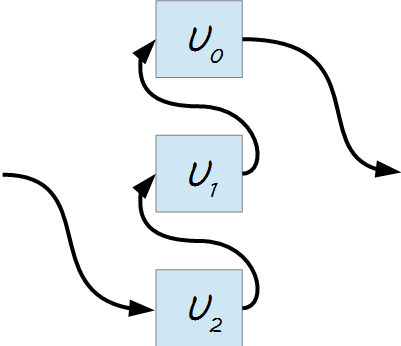}
\subcaption{$\ket{x}_c=\ket{5}$}
\end{subfigure}

\medskip
\begin{subfigure}{\xyz\textwidth}
$\frac{1}{\sqrt{6}}\sum\limits_{x=0}^5\ket{x}_c\ket{\Psi_t}$
\end{subfigure}\hfil
\begin{subfigure}{\xyz\textwidth}
\includegraphics[width=\abcd\textwidth]{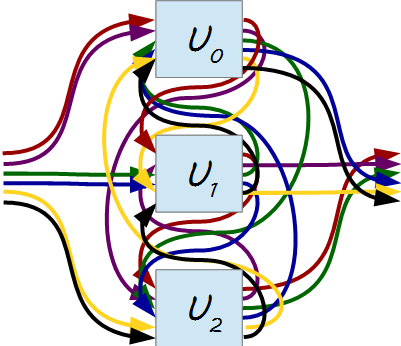}
\end{subfigure}\hfil
\begin{subfigure}{\xyz\textwidth}
$\frac{1}{\sqrt{6}}\sum\limits_{x=0}^5\ket{x}_c \Pi_x\ket{\Psi_t}$
\end{subfigure}\\

    \caption{The quantum-$3$-switch: Depending on the state of the control system, the gates act on the target system in a different order. For the case of $n=3$, each basis state of the six-dimensional control system realizes a different permutation of the gates. If the control system is initialized in a superposition, the $n$-switch can be used to solve Fourier promise problems. In this way, each unitary $U_i$ is called only once.}
    \label{fig:my_label}
\end{figure}

\section{The Fourier Promise Problem} \label{FPP}

In the Fourier promise problem, originally introduced in Ref.~\cite{1}, a set of $d$-dimensional unitary gates $\{U_i\}_0^{n-1}$ is given. Each permutation $\sigma_x$ of the $n$ unitaries is denoted as \mbox{$\Pi_x=U_{\sigma_x(n-1)}...U_{\sigma_x(1)}U_{\sigma_x(0)}$} and labeled by a number \mbox{$x\in \{0,1, ..., n!-1\}$}. It is promised that for some value \mbox{$y\in\{0,1,...,n!-1\}$}, the permutations satisfy the following relation:
\begin{align}
    \forall x\in \{0,1,...,n!-1\}:  \ \Pi_x=\omega^{x\cdot y}\cdot \Pi_0 \, .
    \label{promise}
\end{align}
Here, $\omega$ is defined as $\omega:=e^{\frac{2\pi i}{n!}}$ and the task is to find the value~$y$ for which the above promise is satisfied.\\

For example, in the case of two unitaries $U_0$ and $U_1$, the two permutations $U_1U_0$ and $U_0U_1$ can be either labeled by $\Pi_0=U_1U_0$ and $\Pi_1=U_0U_1$ or the other way around ($\Pi_1=U_1U_0$ and $\Pi_0=U_0U_1$). While the promise for $x=0$ is trivially satisfied, the promise for $x=1$, namely $\Pi_1=\omega^{1\cdot y}\cdot \Pi_0$, translates for both labelings into the fact that $U_0$ and $U_1$ either commute ($y=0$) or anticommute ($y=1$):
\begin{align}
    U_0U_1=(-1)^{y}\cdot U_1U_0 \, .
\end{align}
The task is to find out which property is the correct one.\\

For $n\geq 3$, there are  different ways to label the permutations that lead in general to inequivalent tasks (examples are given in Subsection~\ref{label3}). In this sense, Fourier promise problems form an entire class of problems and we use the term ``specific Fourier promise problem'' whenever we refer to a precise labeling of the permutations. To show that this class of problems is non-trivial, one has to prove that for every $n$ there is at least one specific FPP for which there indeed exist unitaries that satisfy the promise. This is shown in Appendix~A of the original work of M.~Araújo et al.~\cite{1}, where they construct for every $n\geq 2$ and every $y\in \{0,1,...,n!-1\}$ a set of unitaries $\{U_i\}_0^{n-1}$ that satisfy the promise $\Pi_x=\omega^{x\cdot y}\cdot \Pi_0$ for a given labeling of the permutations.\footnote{While the precise form of these unitaries is not relevant in this article (since we will only use the fact that they satisfy the promise), we want to mention that their Hilbert space dimension must be at least $d=n!$. Since these unitaries satisfy the promise, we can consider $\Pi_1=\omega^{y}\cdot \Pi_0$ and take the determinant on both sides:
\begin{align}
    det(\Pi_1)=\omega^{y\cdot d}\cdot det(\Pi_0)\, .
\end{align}
Since $\Pi_0$ and $\Pi_1$ are products of the same unitaries in different order, we obtain $det(\Pi_1)=det(\Pi_0)$ and therefore $\omega^{y\cdot d}=1$. A solution for every $y\in \{0,1,...,n!-1\}$ can only exist if $d\geq n!$.} We want to point out that, for a given $n$, this is the \emph{only} specific FPP for which the existence of these unitaries is explicitly shown (and hence the only task that is proven to be non-trivial). For this specific task, we will present in Section~\ref{seclog} a causal algorithm that is very efficient in the amount of called black-box unitaries (queries), but has the disadvantage that it cannot be adapted directly to (other possibly existing non-trivial) FPPs where the permutations are labeled differently. Note, however, that in Ref.~\cite{1} the distinction between specific FPPs is not made explicitly, since only algorithms are considered that can be adapted to every FPP (independently of the precise labeling of the permutations).

\subsection{Solution with the quantum-$n$-switch}
The quantum-$n$-switch (denoted as $S_n$ and called $n$-switch for short) is the quantum gate that applies, depending on the state of the control system $\ket{x}$, the permutation $\Pi_x$ on the target system~$\ket{\Psi_t}$:
\begin{align}
    \forall x\in\{0,1,...,n!-1\}:\ S_n\ket{x}_c\otimes\ket{\Psi_t}=\ket{x}_c\otimes \Pi_x\ket{\Psi_t} \, .
\end{align}
Moreover, since the $n!$-dimensional Fourier transform is frequently used in this article, we formally introduce it here. In symbols,
\begin{align}
    \forall y\in\{0,1,...,n!-1\}:\ F_{n!}\ket{y}=\frac{1}{\sqrt{n!}}\sum_{x=0}^{n!-1}\omega^{x\cdot y}\ket{x} \, .
\end{align}
With the use of the $n$-switch, one can solve every FPP, as described in Ref.~\cite{1}; the control system is initialized in the $n!$-dimensional state $\ket{0}_c$ and the target system $\ket{\Psi_t}$ in an arbitrary $d$-dimensional state. The Fourier transform $F_{n!}$ transforms the control system into an equal superposition of all states $x\in\{0,1,...,n!-1\}$:
\begin{align}
    \left(F_{n!}\ket{0}_c\right)\otimes \ket{\Psi_t}=\left(\frac{1}{\sqrt{n!}}\sum_{x=0}^{n!-1}\ket{x}_c\right)\otimes\ket{\Psi_t} \, .
\end{align}
Afterwards, the $n$-switch applies, depending on the state $\ket{x}$ of the control system, the permutation $\Pi_x$ on the target system~$\ket{\Psi_t}$ (see Fig.~\ref{fig:my_label} for an illustration of the map for the case of $n=3$):
\begin{align}
    S_n\left(\frac{1}{\sqrt{n!}}\sum_{x=0}^{n!-1}\ket{x}_c\right)\otimes\ket{\Psi_t}=\frac{1}{\sqrt{n!}}\sum_{x=0}^{n!-1}\ket{x}_c\otimes\Pi_x\ket{\Psi_t} \, .
\end{align}
With the promise $\Pi_x=\omega^{x\dot y}\cdot \Pi_0$, this state can be rewritten into:
\begin{align}
    \frac{1}{\sqrt{n!}}\sum_{x=0}^{n!-1}\ket{x}_c\otimes\Pi_x\ket{\Psi_t}=\left(\frac{1}{\sqrt{n!}}\sum_{x=0}^{n!-1}\omega^{x\cdot y}\ket{x}_c\right)\otimes\Pi_0\ket{\Psi_t} \, .
\end{align}
In this way, the target system becomes independent of $x$ and factorizes out in the state $\Pi_0\ket{\Psi_t}$. After applying the inverse Fourier transform on the control system, the desired value of $y$ can be read out with a measurement of the control system in the computational basis:
\begin{align}
    F_{n!}^{-1}\left(\frac{1}{\sqrt{n!}}\sum_{x=0}^{n!-1}\omega^{x\cdot y}\ket{x}_c\right)\otimes\Pi_0\ket{\Psi_t}=\ket{y}_c\otimes\Pi_0\ket{\Psi_t} \, .
\end{align}
Since the $n$-switch can apply every permutation of the unitaries with a single call to each gate, the total query complexity of this algorithm is $n$.

\subsection{Solution with causal quantum algorithms}\label{secsimswitch}

In this section, we give an overview of the best causal algorithms for FPPs that are known. All of them are based on the simulation of the $n$-switch and call $O(n^2)$ black-box unitaries. A causal quantum algorithm simulates the action of the $n$-switch (denoted as $ S_n^{\text{sim.}}$) if it implements the transformation
\begin{align}
\begin{split}
    S_n^{\text{sim.}}&\ket{x}_c\otimes\ket{\Psi_t}\otimes\left(\bigotimes^{n-1}_{i=0}\ket{a_i}\right)\\
    &=\ket{x}_c\otimes\Pi_x\ket{\Psi_t}\otimes\left(\bigotimes^{n-1}_{i=0}(U_i)^{k_i}\ket{a_i}\right)\label{simswitch}
\end{split}
\end{align}
for every $x\in\{0,1,...,n!-1\}$, arbitrary states $\ket{\Psi_t}$ and $\ket{a_i}$ as well as constants $k_i$ that do not depend on $x$. Every simulation of the $n$-switch can be used in combination with the algorithm in Fig.~\ref{figsimswitch} to solve every FPP; analogously to the $n$-switch in the last subsection, the control system is prepared with a quantum Fourier transform in an equal superposition over all states \mbox{$x\in\{0,1,...,n!-1\}$}. By linearity, an algorithm that simulates the $n$-switch implements the transformation:
\begin{align}
\begin{split}
    &S^{\text{sim.}}_n\left(\frac{1}{\sqrt{n!}}\sum_{x=0}^{n!-1}\ket{x}_c\right)\otimes\ket{\Psi_t}\otimes\left(\bigotimes^{n-1}_{i=0}\ket{a_i}\right)\\
    &=\left(\frac{1}{\sqrt{n!}}\sum_{x=0}^{n!-1}\ket{x}_c\otimes\Pi_x\ket{\Psi_t}\right)\otimes\left(\bigotimes^{n-1}_{i=0}(U_i)^{k_i}\ket{a_i}\right)\\
    &=\left(\frac{1}{\sqrt{n!}}\sum_{x=0}^{n!-1}\omega^{x\cdot y}\ket{x}_c\right)\otimes\Pi_0\ket{\Psi_t}\otimes\left(\bigotimes^{n-1}_{i=0}(U_i)^{k_i}\ket{a_i}\right)\, .
\end{split}
\end{align}

Again, the promise $\Pi_x=\omega^{x\cdot y}\cdot \Pi_0$ is used to obtain the last equality. After applying the inverse Fourier transform to the control system, the solution $y$ can be read out in the control system.\\



\newcommand{\te}{S_n^{\text{sim.}}}
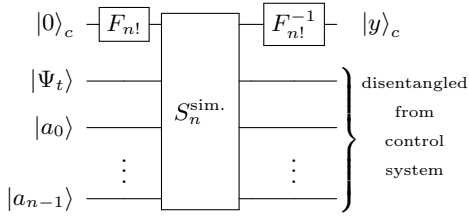
\begin{figure}
    \begin{center}
$\Qcircuit @C=0.5em @R=1em {
\lstick{\ket{0}_c}     & \gate{F_{n!}} & \multigate{4}{\te} & \qw & \gate{F^{-1}_{n!}} & \qw & \rstick{\ket{y}_c}                                                                \\
\lstick{\ket{\Psi_t}}  & \qw             & \ghost{\te}        & \qw & \qw             & \qw &                                                                                     \\
\lstick{\ket{a_0}}     & \qw             & \ghost{\te}        & \qw & \qw             & \qw & \rstick{\substack{\text{disentangled}\\\\\text{from}\\\\\text{control}\\\\\text{system}}} \\
        &    \vdots             &                       &     &  \vdots               &     &                                                                               \\
\lstick{\ket{a_{n-1}}} & \qw             & \ghost{\te}        & \qw & \qw             & \qw \gategroup{2}{5}{5}{6}{1em}{\}} & 
}$
\end{center}
    \caption{Solution of every FPP with the simulation of the quantum-$n$-switch: With a Fourier transform the control system is prepared in an equal superposition of all states $x$. The simulation of the $n$-switch $S_n^\text{sim}$ applies, depending on the state of the control system $\ket{x}$, the permutation $\Pi_x$ on the target system. The solution~$y$ can be read out after applying the inverse Fourier transform to the control system.}
    \label{figsimswitch}
\end{figure}


One algorithm that can implement the transformation $S_n^{\text{sim.}}$ is given in Fig.~\ref{figimplsimswitch}. This one was originally introduced in Ref.~\cite{Colnaghi_2012} and is also presented in Appendix~C of Ref.~\cite{1}. In each step $i=0,1,...,n-1$, the gate $\mathcal{S}$ swaps, controlled on $\ket{x}$, the target system $\ket{\Psi_t}$ with the auxiliary system $\ket{a_{\sigma_x(i)}}$. After the gate $U_{\sigma_x(i)}$ acts on $\ket{\Psi_t}$, another gate $\mathcal{S}$ swaps the two systems $\ket{\Psi_t}$ and $\ket{a_{\sigma_x(i)}}$ back. In this way, the permutation \mbox{$\Pi_x=U_{\sigma_x(n-1)}...U_{\sigma_x(1)}U_{\sigma_x(0)}$} is applied to the target system. In this algorithm, each gate $U_i$ is used exactly $n$ times, so the query complexity of this algorithm is $n^2$. Furthermore, for every permutation $\Pi_x$, each auxiliary system $\ket{a_i}$ is swapped back and forth with $\ket{\Psi_t}$ exactly once. Hence, independently of the state of the control system $\ket{x}$, the gate $U_i$ is applied once on $\ket{\Psi_t}$ and the remaining $(n-1)$ times on $\ket{a_i}$. In this way, each auxiliary system $\ket{a_i}$ ends up in the state $(U_i)^{n-1}\ket{a_i}$ and the algorithm implements the transformation given in Eq.~\eqref{simswitch} for $k_i=(n-1)$.  \footnote{Further details about the representation of the control system $\ket{x}_c$ and the implementation of the $\mathcal{S}$-gate can be found in Appendix~C of Ref.~\cite{1} and in Ref.~\cite{Colnaghi_2012}.}\\

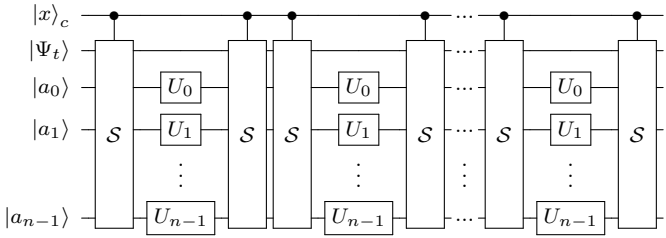
\begin{figure}[h!]

\smaller[1]
\newcommand{\ze}{7}
\begin{flushright}
$\Qcircuit @C=0.3em @R=0.5em {
\lstick{\ket{x}_c} & \qw & \ctrl{2} & \qw & \qw & \qw & \ctrl{2} &\ctrl{2} & \qw & \qw & \qw & \ctrl{2} & \qw && \text{...}&& & \ctrl{2} & \qw & \qw & \qw & \ctrl{2} & \qw \\\\
\lstick{\ket{\Psi_t}} & \qw & \multigate{\ze}{\mathcal{S}} & \qw &\qw &\qw& \multigate{\ze}{\mathcal{S}} & \multigate{\ze}{\mathcal{S}} & \qw &\qw &\qw& \multigate{\ze}{\mathcal{S}} & \qw &&\text{...}&&& \multigate{\ze}{\mathcal{S}} & \qw &\qw &\qw& \multigate{\ze}{\mathcal{S}} &\qw \\
\lstick{\ket{a_0}} & \qw & \ghost{\mathcal{S}}        & \qw &\gate{U_0}      &\qw& \ghost{\mathcal{S}}        & \ghost{\mathcal{S}}        & \qw &\gate{U_0}      &\qw& \ghost{\mathcal{S}}        & \qw &&\text{...}&&   & \ghost{\mathcal{S}}        & \qw &\gate{U_0}      &\qw& \ghost{\mathcal{S}}        & \qw \\
\lstick{\ket{a_1}} & \qw & \ghost{\mathcal{S}}        & \qw &\gate{U_1}      &\qw& \ghost{\mathcal{S}}        & \ghost{\mathcal{S}}        & \qw &\gate{U_1}      &\qw& \ghost{\mathcal{S}}        & \qw &&\text{...}&&   & \ghost{\mathcal{S}}        & \qw &\gate{U_1}      &\qw& \ghost{\mathcal{S}}        & \qw \\\\
&&&&\vdots&&&&&\vdots&&&&&\vdots&&&&&\vdots&&&&\\\\\\
\lstick{\ket{a_{n-1}}} & \qw & \ghost{\mathcal{S}}        & \qw &\gate{U_{n-1}}      &\qw& \ghost{\mathcal{S}}        & \ghost{\mathcal{S}}        & \qw &\gate{U_{n-1}}      &\qw& \ghost{\mathcal{S}}        & \qw &&  \text{...}&&& \ghost{\mathcal{S}}        & \qw &\gate{U_{n-1}}      &\qw& \ghost{\mathcal{S}}        & \qw 
}$
\end{flushright}
\normalsize
    \caption{A simulation of the $n$-switch with a causal algorithm: The permutation $\Pi_x$ is applied on $\ket{\Psi_t}$ by swapping the target system in each step $i=0,1,...,n-1$ with the auxiliary system $\ket{a_{\sigma_x(i)}}$.}
    \label{figimplsimswitch}
\end{figure}

There are simulations of the quantum-$n$-switch with causal quantum circuits that are slightly more efficient. All of them call $O(n^2)$ black-box gates in total. This is studied in Ref.~\cite{Facchini_2015} (and also in Ref.~\cite{1}).

\section{Towards more efficient causal algorithms}\label{sectowards}
\subsection{Pairwise commutation relations}\label{secpair}

In this article, we show that there are algorithms that solve Fourier promise problems by calling significantly less queries than the simulation of the $n$-switch requires. The main ingredient is a property of the unitaries that can be directly inferred from the promise.

\begin{definition}[Pairwise commutation relations]
A set of unitaries $\{U_i\}_0^{n-1}$ satisfy \emph{``pairwise commutation relations''}, if for every pair of unitaries $U_j$ and $U_k$ (\mbox{$j,k\in\{0,1,...,n-1\}$}) there exist $\alpha_{jk}\in \mathbb{C}$ such that:
\begin{align}
\begin{split}
    U_jU_k=\alpha_{jk} \cdot U_kU_j \, .
\end{split}
\end{align}
\end{definition}

\begin{proposition}\label{proposition1}
Every set of unitaries $\{U_i\}_0^{n-1}$ that satisfy the promise of a Fourier promise problem, 
\begin{align}
    \forall x\in &\{0,1,...,n!-1\}: \Pi_x=\omega^{x\cdot y}\cdot \Pi_0 \, ,
\end{align}
satisfy pairwise commutation relations. Furthermore, if for a specific FPP the labeling of the permutations is given, the pairwise commutation relations read:
\begin{align}
\begin{split}
    U_jU_k&=\omega^{(x^1_{jk}-x^2_{jk})\cdot y}\cdot U_kU_j \, ,
    \label{pair}
\end{split}
\end{align}
where $x^1_{jk}$ is the label of the permutation \mbox{$\Pi_{x^1_{jk}}=U_{n-1}...U_{1}U_{0}U_jU_k$} and $x^2_{jk}$ is the label of the permutation \mbox{$\Pi_{x^2_{jk}}=U_{n-1}...U_{1}U_{0}U_kU_j$} (with $U_{n-1}...U_{1}U_{0}$, we denote all unitaries without $U_j$ and $U_k$ in descending order). \footnote{Note that with some caution, one could turn this statement into an equivalence: the condition of pairwise commutation relations is not only a necessary condition for the promise to hold. One can also check by direct calculation that whenever all permutation relations are pairwise \mbox{($\forall j\neq k\ \exists \ x_{jk}\in \mathbb{Z} \text{ s.t. } U_jU_k=\omega^{x_{jk}\cdot y}\cdot U_kU_j$)}, all permutations of the $n$ unitaries are related by a phase \mbox{($\Pi_x=\omega^{x\cdot y}\ \Pi_0$)}. In order for the promise to hold, one has to choose the pairwise phases $x_{jk}$ such that every $x\in\{0,1,...,n!-1\}$ appears exactly once. Nevertheless, it is enough for our purpose that the promise induces pairwise commutation relations.}

\end{proposition}

\begin{proof}
For every pair of black-box unitaries $U_j$ and $U_k$, we focus on the two permutations introduced in the statement above:
\begin{align}
    \Pi_{x^1_{jk}}&=U_{n-1}...U_{1}U_{0}U_jU_k \text{ and }\\
    \Pi_{x^2_{jk}}&=U_{n-1}...U_{1}U_{0}U_kU_j \, .
\end{align}
Due to the promise $\Pi_x=\omega^{x\cdot y}\ \Pi_0$, both permutations are equal to $\Pi_0$ up to the phases $\omega^{-x^1_{jk}\cdot y}$ and $\omega^{-x^2_{jk}\cdot y}$, respectively:
\begin{align}
    \omega^{-x^1_{jk}\cdot y}\cdot \Pi_{x^1_{jk}}=\Pi_0=\omega^{-x^2_{jk}\cdot y}\cdot \Pi_{x^2_{jk}} \, .
\end{align}
Using the above expression for the two permutations $\Pi_{x^1_{jk}}$ and $\Pi_{x^2_{jk}}$, we obtain
\begin{align}
\begin{split}
    \omega^{-x^1_{jk}\cdot y}\cdot U_{n-1}...U_{0}U_jU_k&=\omega^{-x^2_{jk}\cdot y}\cdot U_{n-1}...U_{0}U_kU_j \, .
\end{split}
\end{align}
Multiplying from the left step by step with the inverses of $U_{n-1}$, $U_{n-2}$, ..., $U_{1}$ and $U_{0}$ (for $U_j$ and $U_k$ this is left out), this expression is equivalent to:
\begin{align}
\begin{split}
    U_jU_k&=\omega^{(x^1_{jk}-x^2_{jk})\cdot y}\cdot U_kU_j \, .  
\end{split}
\end{align}
With \mbox{$\alpha_{jk}:=\omega^{(x^1_{jk}-x^2_{jk})\cdot y}\in \mathbb{C}$}, we conclude that every set of unitaries that satisfies the promise also satisfies pairwise commutation relations.
\end{proof}

\subsection{Structure of the new algorithms}\label{newmethod}

\begin{figure}[h!]
\newcommand{\tex}{T_n^{\text{FPP}}}
\centering
$\Qcircuit @C=0.5em @R=1em {
\lstick{\ket{0}_c}     & \gate{F_{n!}} & \multigate{6}{\tex} & \qw & \gate{F^{-1}_{n!}} & \qw & \rstick{\ket{y}_c}                                                                  \\
\lstick{\ket{\Psi_1}}  & \qw             & \ghost{\tex}        & \qw & \qw             & \qw &                                                                                     \\
      &   \vdots              &                       &     &   \vdots              &     &                                                                                     \\
\lstick{\ket{\Psi_m}}  & \qw             & \ghost{\tex}        & \qw & \qw             & \qw & \rstick{\substack{\text{disentangled}\\\\\text{from}\\\\\text{control}\\\\\text{system}}} \\
\lstick{\ket{a_0}}     & \qw             & \ghost{\tex}        & \qw & \qw             & \qw &                                                                                     \\
       & \vdots                &                       &     &  \vdots               &     &                                                                                     \\
\lstick{\ket{a_{n-1}}} & \qw             & \ghost{\tex}        & \qw & \qw             & \qw \gategroup{2}{6}{7}{6}{1em}{\}} &        
}$
    \caption{General structure of the algorithms in this article: With a quantum Fourier transform the control system is initialized in an equal superposition of all states $\ket{x}$. After the transformation $T_n^{\text{FPP}}$ (see main text) is applied, the final state of each target and auxiliary system becomes independent of $x$ and the solution~$y$ can be read out by a measurement of the control system in the Fourier basis.}
    \label{figgeneral}
\end{figure}
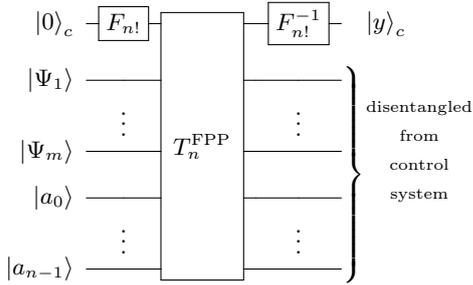
All causal algorithms that we present in this article are of the form given in Fig.~\ref{figgeneral}. The target systems $\ket{\Psi_j}$ and auxiliary systems $\ket{a_i}$ are initialized in an arbitrary $d$-dimensional state. The important part of the algorithm is the one that realizes the transformation $T_n^{\text{FPP}}$:
\begin{align}
\begin{split}
    &T_n^{\text{FPP}}\ket{x}_c\otimes\bigotimes^{m}_{j=1}\ket{\Psi_j}\otimes\bigotimes^{n-1}_{i=0}\ket{a_i}\\
    &=\ket{x}_c\otimes\bigotimes^{m}_{j=1}f^x_j(U_0,..., U_{n-1})\ket{\Psi_j}\otimes\bigotimes^{n-1}_{i=0}(U_i)^{k_i}\ket{a_i}\\
    &=\omega^{x\cdot y}\ket{x}_c\otimes\bigotimes^{m}_{j=1}f^0_j(U_0,..., U_{n-1})\ket{\Psi_j}\otimes\bigotimes^{n-1}_{i=0}(U_i)^{k_i}\ket{a_i} \, . \label{deftfpp}
\end{split}
\end{align}
On each target system $\ket{\Psi_j}$, some of the $n$ black-box unitaries from the set are applied, such that each of them ends up in a state $f^x_j(U_0,..., U_{n-1})\ket{\Psi_j}$. The unitaries contained within $f^x_j(U_0,..., U_{n-1})$ are the same for every $x$, but the order in which these unitaries are applied on the target system will explicitly depend on $x$. Since all commutation relations are pairwise, one can always rewrite this expression into $f^0_j(U_0,..., U_{n-1})\ket{\Psi_j}$ and a phase is obtained whenever two unitaries are commuted (with $e^{i\phi_j(x)}$ we denote the product of these phases):
\begin{align}
    f^x_j(U_0,..., U_{n-1})\ket{\Psi_j}=e^{i\phi_j(x)}\  f^0_j(U_0,..., U_{n-1})\ket{\Psi_j} \, .
\end{align}
In this way, the final state of each target system becomes independent of $x$ and if the algorithm is designed carefully, all these phases multiply together to $\omega^{x\cdot y}$. \\

Whenever we find an implementation of this transformation $T_n^{\text{FPP}}$ for every $x\in\{0,1,...,n!-1\}$ we can solve the corresponding FPP. The control system is initialized in an equal superposition of all $x$ and due to linearity, the transformation $T_n^{\text{FPP}}$ realizes:
\begin{align}
\begin{split}
    &T_n^{\text{FPP}}\left(\frac{1}{\sqrt{n!}}\sum_{x=0}^{n!-1}\ket{x}_c\right)\otimes\left(\bigotimes^{m}_{j=1}\ket{\Psi_j}\right)\otimes\left(\bigotimes^{n-1}_{i=0}\ket{a_i}\right)\\
    &=\left(\frac{1}{\sqrt{n!}}\sum_{x=0}^{n!-1}\omega^{x\cdot y}\ket{x}_c\right)\otimes\left(\bigotimes^{m}_{j=1}f^0_j(U_0,..., U_{n-1})\ket{\Psi_j}\right)\\
    &\hspace{2cm}\otimes\left(\bigotimes^{n-1}_{i=0}(U_i)^{k_i}\ket{a_i}\right) \, .
\end{split}
\end{align}
At the end, the solution~$y$ can be read out after applying the inverse Fourier transform to the control system (see Fig.~\ref{figgeneral}).\\

Intuitively speaking, the pairwise commutation relations allow us to simulate different parts of the total phase $\omega^{x\cdot y}$ on different target systems. Note that the best causal algorithms so far that are based on the simulation of the $n$-switch (Subsection~\ref{secsimswitch}) can be seen as a special case of this new method with only one target system ($m=1$) and $f^x_1(U_0,..., U_{n-1})=\Pi_x$. Hence, this new procedure is more general and usually more efficient than simulating every permutation on its own. We want to point out that these ideas can in principle be applied to every set of unitaries satisfying pairwise commutation relations. In this sense, this method may have some applications beyond Fourier promise problems.\\


\section{More efficient solutions for FPPs with three unitaries}\label{secexample3}
In this section, we show in a first example how pairwise commutation relations are useful to find more efficient causal algorithms. More precisely, we present an algorithm that solves every FPP for $n=3$ with six queries, while the best causal algorithm that was known so far used seven queries. While this difference might not seem significant at first, we will show in the next sections that similar ideas can be used for a significant reduction in the number of used black-box gates in the asymptotic limit.\\

\subsection{Two possible ways to label the permutations}\label{label3}
Before we present the algorithms, we give two specific examples of labelings for $n=3$. We denote with \mbox{$x_{ijk}\in \{0,1,...,3!-1\}$} the label of the permutation $U_iU_jU_k$: 
\begin{align}
    0=x_{210}&&\Pi_0=U_2U_1U_0&\label{fac31} \\
    1=x_{201}&&\Pi_1=U_2U_0U_1&=\omega^{1\cdot y}\ U_2U_1U_0\\
    2=x_{120}&&\Pi_2=U_1U_2U_0&=\omega^{2\cdot y}\ U_2U_1U_0\\
    3=x_{021}&&\Pi_3=U_0U_2U_1&=\omega^{3\cdot y}\ U_2U_1U_0\\
    4=x_{102}&&\Pi_4=U_1U_0U_2&=\omega^{4\cdot y}\ U_2U_1U_0\\
    5=x_{012}&&\Pi_5=U_0U_1U_2&=\omega^{5\cdot y}\ U_2U_1U_0 \, . \label{fac3}
\end{align}
If it is promised that the three unitaries satisfy these conditions, it is straightforward to read off the pairwise commutation relations: from the second line $U_2U_0U_1=\omega^{1\cdot y}\ U_2U_1U_0$, it follows that $U_0U_1=\omega^{1\cdot y}\ U_1U_0$, and from the third line, one can conclude that $U_1U_2=\omega^{2\cdot y}\ U_2U_1$. For the last pair $U_0$ and $U_2$, we have to put in some more work: Due to Proposition~\ref{proposition1}, we can compare the fifth ($U_1U_0U_2$) and the third line ($U_1U_2U_0$) to obtain:
\begin{align}
    \begin{split}
       \omega^{-4\cdot y}\ U_1&U_0U_2=U_2U_1U_0=\omega^{-2\cdot y}\ U_1U_2U_0\\
       &\implies U_0U_2=\omega^{2\cdot y}\ U_2U_0.
    \end{split}
\end{align}


There exist other FPPs that correspond to other labelings. Another example is the following:
\begin{align}
    0=x_{210}&&\Pi_0=U_2U_1U_0&\\
    1=x_{102}&&\Pi_1=U_1U_0U_2&=\omega^{1\cdot y}\ U_2U_1U_0\\
    2=x_{201}&&\Pi_2=U_2U_0U_1&=\omega^{2\cdot y}\ U_2U_1U_0\\
    3=x_{012}&&\Pi_3=U_0U_1U_2&=\omega^{3\cdot y}\ U_2U_1U_0\\
    4=x_{120}&&\Pi_4=U_1U_2U_0&=\omega^{4\cdot y}\ U_2U_1U_0\\
    5=x_{021}&&\Pi_5=U_0U_2U_1&=\omega^{5\cdot y}\ U_2U_1U_0 \, .
\end{align}
Here, the pairwise commutation relations can be read off as follows: $U_0U_1=\omega^{2\cdot y}\ U_1U_0$, $U_0U_2=\omega^{3\cdot y}\ U_2U_0$ and $U_1U_2=\omega^{4\cdot y}\ U_2U_1$. On the other hand, knowing all pairwise phases uniquely determines the labeling (up to the freedom of choosing $\Pi_0$).\\

Note that not every labeling is meaningful. Some of them lead to trivial statements:
\begin{align}
    0=x_{210}&&\Pi_0=U_2U_1U_0&\\
    1=x_{201}&&\Pi_1=U_2U_0U_1&=\omega^{1\cdot y}\ U_2U_1U_0\\
    2=x_{102}&&\Pi_2=U_1U_0U_2&=\omega^{2\cdot y}\ U_2U_1U_0\\
    3=x_{120}&&\Pi_3=U_1U_2U_0&=\omega^{3\cdot y}\ U_2U_1U_0\\
    4=x_{012}&&\Pi_4=U_0U_1U_2&=\omega^{4\cdot y}\ U_2U_1U_0\\
    5=x_{021}&&\Pi_5=U_0U_2U_1&=\omega^{5\cdot y}\ U_2U_1U_0 \, .
\end{align}
From the second line, it follows $U_0U_1=\omega^{1\cdot y}\ U_1U_0$, while from comparing $\Pi_2=U_1U_0U_2$ with $\Pi_4=U_0U_1U_2$, one obtains $U_0U_1=\omega^{2\cdot y}\ U_1U_0$. This is a contradiction whenever $y\neq 0$ (note that $\omega=e^{\frac{2\pi i}{n!}} \neq 1$ for $n\geq 2$). More precisely, only for $y=0$, there exist unitaries $U_0$, $U_1$ and $U_2$ that satisfy this promise, and the task becomes trivial (since one can conclude directly that the solution must be $y=0$). By counting, we found that there are 24 different possible labelings of the six permutations that lead to non-trivial solutions for $n=3$ if we restrict ourselves to $x_{210}=0$ ($\Pi_0=U_2U_1U_0$).

\subsection{Standard causal algorithm with seven queries}\label{secexample37}
The best causal algorithms known so far that solve these problems are based on the simulation of the 3-switch and call seven black-box unitaries. One possible algorithm that can achieve this is given in Fig.~\ref{figex37}. The gates $\mathcal{R}$ denote rewirings of the target and auxiliary systems (a combination of controlled swaps). Depending on the state of the control system $\ket{x_{ijk}}$, they interchange the wires in a way that the gates act on the systems according to Table~\ref{tabex37}. All underlined gates \underline{$U_i$} act on $\ket{\Psi_t}$ and simulate the permutation $U_iU_jU_k$, while the remaining (unused) gates $U_0$ and $U_1$ act on the corresponding auxiliary systems $\ket{a_0}$ and $\ket{a_1}$, respectively.

\begin{figure}[h!]
\centering
\smaller[2]
\begin{flushright}
$\Qcircuit @C=0.1em @R=0.7em {
\lstick{\ket{x}_c}& \qw & \ctrl{1}& \qw & \qw & \qw & \ctrl{1}& \qw & \qw & \qw & \ctrl{1} & \qw & \qw & \qw & \ctrl{1}  & \qw & \qw & \qw & \ctrl{1} & \qw & \qw & \qw & \ctrl{1}                   & \qw & \qw & \qw & \ctrl{1} & \qw & \qw & \qw & \ctrl{1}                   & \qw \\
\lstick{\ket{\Psi_{t}}} & \qw & \multigate{2}{\mathcal{R}} & \qw &\gate{U_1} &\qw& \multigate{2}{\mathcal{R}} & \qw &\gate{U_0} &\qw& \multigate{2}{\mathcal{R}} & \qw &\gate{U_1} &\qw& \multigate{2}{\mathcal{R}} & \qw &\gate{U_2} &\qw& \multigate{2}{\mathcal{R}} & \qw &\gate{U_1} &\qw& \multigate{2}{\mathcal{R}} & \qw &\gate{U_0} &\qw& \multigate{2}{\mathcal{R}} & \qw &\gate{U_1} &\qw& \multigate{2}{\mathcal{R}} & \qw\\
\lstick{\ket{a_0}} & \qw & \ghost{\mathcal{R}}        & \qw &\qw      & \qw & \ghost{\mathcal{R}}        & \qw &\qw      &\qw& \ghost{\mathcal{R}}        & \qw &\qw      &\qw& \ghost{\mathcal{R}}        & \qw &\qw      &\qw& \ghost{\mathcal{R}}        & \qw &\qw      &\qw& \ghost{\mathcal{R}}        & \qw &\qw      &\qw& \ghost{\mathcal{R}}        & \qw &\qw      &\qw& \ghost{\mathcal{R}}        & \qw\\
\lstick{\ket{a_1}} & \qw & \ghost{\mathcal{R}}        & \qw &\qw      & \qw & \ghost{\mathcal{R}}        & \qw &\qw      &\qw& \ghost{\mathcal{R}}        & \qw &\qw      &\qw& \ghost{\mathcal{R}}        & \qw &\qw      &\qw& \ghost{\mathcal{R}}        & \qw &\qw      &\qw& \ghost{\mathcal{R}}        & \qw &\qw      &\qw& \ghost{\mathcal{R}}        & \qw &\qw      &\qw& \ghost{\mathcal{R}}        & \qw 
}$
\end{flushright}
\normalsize
    \caption{Simulation of the $3$-switch ($S^\text{sim.}_3$) with the smallest possible number of used black-box gates.}
    \label{figex37}
\end{figure}
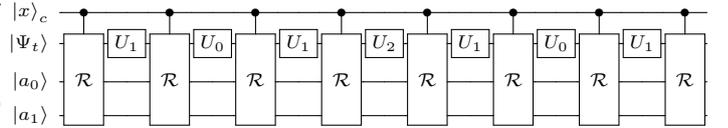

\begin{table}[h!]
    \centering
    \begin{tabular}{l|l}
    &$S^\text{sim.}_3\left(\ket{x_{ijk}}\otimes \ket{\Psi_t}\otimes \ket{a_0}\otimes \ket{a_1}\right)$\\\hline
    $U_1$\underline{$U_0U_1U_2$}$U_1U_0U_1$
    &$\ket{x_{210}}\otimes \underline{U_2U_1U_0} \ket{\Psi_{t}}\otimes U_0 \ket{a_0}\otimes (U_1)^3 \ket{a_1}$
    \\
    \underline{$U_1U_0$}$U_1$\underline{$U_2$}\textcolor{black}{$U_1U_0$}\textcolor{black}{$U_1$}
    &$\ket{x_{201}}\otimes \underline{U_2U_0U_1}\ket{\Psi_{t}}\otimes U_0 \ket{a_0}\otimes (U_1)^3 \ket{a_1}$
    \\
    $U_1$\underline{$U_0$}\textcolor{black}{$U_1$}\underline{$U_2U_1$}\textcolor{black}{$U_0U_1$} 
    &$\ket{x_{120}}\otimes \underline{U_1U_2U_0} \ket{\Psi_{t}}\otimes U_0 \ket{a_0}\otimes (U_1)^3 \ket{a_1}$
    \\
    $U_1$\textcolor{black}{$U_0$}\underline{$U_1$}\underline{$U_2$}\textcolor{black}{$U_1$}\underline{$U_0$}\textcolor{black}{$U_1$} 
    &$\ket{x_{021}}\otimes  \underline{U_0U_2U_1} \ket{\Psi_{t}}\otimes U_0 \ket{a_0}\otimes (U_1)^3 \ket{a_1}$
    \\
    $U_1$\textcolor{black}{$U_0$}\textcolor{black}{$U_1$}\underline{$U_2$}\textcolor{black}{$U_1$}\underline{$U_0$}\underline{$U_1$} 
    &$\ket{x_{102}}\otimes  \underline{U_1U_0U_2} \ket{\Psi_{t}}\otimes U_0 \ket{a_0}\otimes (U_1)^3 \ket{a_1}$
    \\
    $U_1$\textcolor{black}{$U_0$}\textcolor{black}{$U_1$}\underline{$U_2$}\underline{$U_1$}\underline{$U_0$}\textcolor{black}{$U_1$} 
    &$\ket{x_{012}}\otimes \underline{U_0U_1U_2} \ket{\Psi_{t}}\otimes U_0 \ket{a_0}\otimes (U_1)^3 \ket{a_1}$
    \end{tabular}
    \caption{Final state for every $\ket{x_{ijk}}$ of the algorithm in Fig.~\ref{figex37}.}
    \label{tabex37}
\end{table}

In total, the combined control and target system simulate the action of the 3-switch, since every permutation of the three unitaries can be applied on the target system, while the two auxiliary systems always end up in the same state
\begin{align}
\begin{split}
    &S^\text{sim.}_3\ \ket{x_{ijk}}\otimes \ket{\Psi_t}\otimes \ket{a_0}\otimes \ket{a_1}\\
    &=\ket{x_{ijk}}\otimes U_iU_jU_k\ket{\Psi_t}\otimes U_0 \ket{a_0}\otimes (U_1)^3 \ket{a_1}
    \, .  \label{simswitch3}
\end{split}
\end{align}
As explained in Subsection~\ref{secsimswitch}, this solves every FPP as the 3-switch itself.

It is essential for an algorithm that simulates the $n$-switch that every permutation of the unitaries can be applied on the target system by rewiring the systems in some way. Here, for $n=3$, every permutation has to appear as a substring in $U_1U_0U_1U_2U_1U_0U_1$. The minimal length of a string of elements $U_0, U_1, ..., U_{n-1}$ such that all possible permutations of the $n$ elements are contained in the string as a substring is a well-studied problem in combinatorics: it is known that the minimal number of elements in such a string is of the order of $O(n^2)$ \cite{Kleitman_1976}. For $3\leq n \leq 7$, the shortest string containing all permutations as a substring has length $n^2-2n+4$ \cite{newey}. For higher $n$, more efficient constructions are known \cite{zalinescu, radomirovic}. For this reason, no string of length smaller than seven can contain all the six permutations of three unitaries. 

\subsection{More efficient causal algorithm with six queries}\label{secdrei}
Here, we will show that we can solve every FPP for $n=3$ with only six queries using the algorithm given in Fig.~\ref{figex36}. As before, the gates $\mathcal{R}$ are rewirings of the target and auxiliary systems. They interchange the wires in a way that the gates act on the systems according to Tab.~\ref{tabex36}. All underlined gates \underline{$U_i$} act on $\ket{\Psi_1}$, all overlined gates $\overline{U_i}$ on $\ket{\Psi_2}$ and the remaining gate $U_1$ acts on the auxiliary system $\ket{a_1}$.

\begin{figure}[h!]
\smaller[1]
\begin{flushright}
$\Qcircuit @C=0.1em @R=0.7em {
\lstick{\ket{x}_c}& \qw & \ctrl{1}& \qw & \qw & \qw & \ctrl{1} & \qw & \qw & \qw & \ctrl{1}  & \qw & \qw & \qw & \ctrl{1} & \qw & \qw & \qw & \ctrl{1}                   & \qw & \qw & \qw & \ctrl{1} & \qw & \qw & \qw & \ctrl{1}                   & \qw \\
\lstick{\ket{\Psi_{1}}} & \qw & \multigate{2}{\mathcal{R}} & \qw &\gate{U_0} &\qw& \multigate{2}{\mathcal{R}} & \qw &\gate{U_1} &\qw& \multigate{2}{\mathcal{R}} & \qw &\gate{U_2} &\qw& \multigate{2}{\mathcal{R}} & \qw &\gate{U_1} &\qw& \multigate{2}{\mathcal{R}} & \qw &\gate{U_0} &\qw& \multigate{2}{\mathcal{R}} & \qw &\gate{U_1} &\qw& \multigate{2}{\mathcal{R}} & \qw\\
\lstick{\ket{\Psi_{2}}} & \qw & \ghost{\mathcal{R}}        & \qw &\qw      &\qw& \ghost{\mathcal{R}}        & \qw &\qw      &\qw& \ghost{\mathcal{R}}        & \qw &\qw      &\qw& \ghost{\mathcal{R}}        & \qw &\qw      &\qw& \ghost{\mathcal{R}}        & \qw &\qw      &\qw& \ghost{\mathcal{R}}        & \qw &\qw      &\qw& \ghost{\mathcal{R}}        & \qw\\
\lstick{\ket{a_1}} & \qw & \ghost{\mathcal{R}}        & \qw &\qw      &\qw& \ghost{\mathcal{R}}        & \qw &\qw      &\qw& \ghost{\mathcal{R}}        & \qw &\qw      &\qw& \ghost{\mathcal{R}}        & \qw &\qw      &\qw& \ghost{\mathcal{R}}        & \qw &\qw      &\qw& \ghost{\mathcal{R}}        & \qw &\qw      &\qw& \ghost{\mathcal{R}}        & \qw 
}$
\end{flushright}
\caption{An implementation of $T_3^\text{FPP}$ that solves every FPP for three unitaries.}
    \label{figex36}
\end{figure}
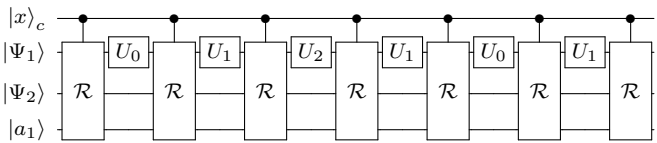

\normalsize

\begin{table}[H]
    \centering
    \begin{tabular}{l|l}
    &$T_3^\text{FPP}\left(\ket{x_{ijk}}\otimes \ket{\Psi_1}\otimes \ket{\Psi_2}\otimes \ket{a_1}\right)$\\\hline
    \underline{$U_0U_1U_2$}\textcolor{black}{$U_1$}$\overline{U_0U_1}$
    &$\ket{x_{210}}\otimes \underline{U_2U_1U_0} \ket{\Psi_{1}}\otimes\overline{U_1U_0}\ket{\Psi_{2}}\otimes U_1 \ket{a_1}$
    \\
    \underline{$U_0U_1U_2$}$\overline{U_1U_0}$\textcolor{black}{$U_1$}
    &$\ket{x_{201}}\otimes \underline{U_2U_1U_0}\ket{\Psi_{1}}\otimes \overline{U_0U_1}\ket{\Psi_{2}}\otimes U_1\ket{a_1}$
    \\
    &
    $=\ket{x_{201}}\otimes U_2U_0U_1\ket{\Psi_{1}}\otimes U_1U_0\ket{\Psi_{2}}\otimes U_1\ket{a_1}$
    \\
    \underline{$U_0$}\textcolor{black}{$U_1$}\underline{$U_2U_1$}$\overline{U_0U_1}$ 
    &$\ket{x_{120}}\otimes \underline{U_1U_2U_0} \ket{\Psi_{1}}\otimes\overline{U_1U_0}\ket{\Psi_{2}}\otimes  U_1 \ket{a_1}$
    \\
    $\overline{U_0}$\underline{$U_1$}\underline{$U_2$}\textcolor{black}{$U_1$}\underline{$U_0$}$\overline{U_1}$ 
    &$\ket{x_{021}}\otimes  \underline{U_0U_2U_1} \ket{\Psi_{1}}\otimes\overline{U_1U_0}\ket{\Psi_{2}}\otimes  U_1 \ket{a_1}$
    \\
    $\overline{U_0}\overline{U_1}$\underline{$U_2$}\textcolor{black}{$U_1$}\underline{$U_0$}\underline{$U_1$} 
    &$\ket{x_{102}}\otimes  \underline{U_1U_0U_2} \ket{\Psi_{1}}\otimes\overline{U_1U_0}\ket{\Psi_{2}}\otimes  U_1 \ket{a_1}$
    \\
    $\overline{U_0}\overline{U_1}$\underline{$U_2$}\underline{$U_1$}\underline{$U_0$}\textcolor{black}{$U_1$} 
    & $\ket{x_{012}}\otimes \underline{U_0U_1U_2} \ket{\Psi_{1}}\otimes\overline{U_1U_0}\ket{\Psi_{2}}\otimes  U_1 \ket{a_1}$
    \end{tabular}
    \caption{Final state for every $\ket{x_{ijk}}$ of the algorithm in Fig.~\ref{figex36}. Using pairwise commutation relations, one can rewrite the second line (see main text).}
    \label{tabex36}
\end{table}

The crucial difference here is the permutation $U_2U_0U_1$. It is not possible that the first target system $\ket{\Psi_1}$ ends up in the state $U_2U_0U_1\ket{\Psi_1}$ directly, since the permutation $U_1U_0U_2$ is not contained as a substring of $U_0U_1U_2U_1U_0U_1$ (remember that the order is reversed since $U_1$ has to act first, then $U_0$ and $U_2$). In this way, this algorithm is not able to simulate the (complete) 3-switch but nevertheless, it is able to solve FPPs. Since every set of unitaries that satisfies the promise satisfies pairwise commutation relations, we can use \mbox{$U_1U_0=\alpha_{10}\ U_0U_1$} to rewrite the second line of Table~\ref{tabex36} into:
\begin{align}
\begin{split}
    &\ket{x_{201}}\otimes U_2U_1U_0\ket{\Psi_{1}}\otimes U_0U_1\ket{\Psi_{2}}\otimes U_1\ket{a_1}\\
    &=\alpha_{10}\ \ket{x_{201}}\otimes U_2U_0U_1\ket{\Psi_{1}}\otimes U_0U_1\ket{\Psi_{2}}\otimes U_1\ket{a_1}\\
    &=\ket{x_{201}}\otimes U_2U_0U_1\ket{\Psi_{1}}\otimes U_1U_0\ket{\Psi_{2}}\otimes U_1\ket{a_1}\, .
\end{split}
\end{align}
Due to this, the algorithm implements a transformation that is very similar to the one that is implemented by the algorithm in the last subsection (Eq.~\eqref{simswitch3}). For every $x_{ijk}$ the permutation $U_iU_jU_k$ is applied on the first target system, while the second target system and the auxiliary system always end up in the state $U_1U_0\ket{\Psi_{2}}\otimes U_1\ket{a_1}$:
\begin{align}
\begin{split}
    &T_3^\text{FPP}\ \ket{x_{ijk}}\otimes \ket{\Psi_1}\otimes \ket{\Psi_2}\otimes \ket{a_1}\\
    &=\ket{x_{ijk}}\otimes U_iU_jU_k\ket{\Psi_1}\otimes U_1U_0\ket{\Psi_{2}}\otimes U_1\ket{a_1} \label{tfpp3}
    \, .
\end{split}
\end{align}
In this way, the combined control and first target system simulate the action of the 3-switch for unitaries satisfying pairwise commutation relations. Therefore, this algorithm solves every FPP for $n=3$ in the same way as the algorithm in the last subsection.\footnote{Alternatively, to stay within the methods developed in Subsection~\ref{newmethod}, note that Eq.~\eqref{tfpp3} can be rewritten into $\omega^{x_{ijk}\cdot y}\ \ket{x_{ijk}}\otimes \Pi_0\ket{\Psi_1}\otimes U_1U_0\ket{\Psi_{2}}\otimes U_1\ket{a_1}$ by using the promise $\Pi_{x_{ijk}}=U_iU_jU_k=\omega^{x_{ijk}\cdot y}\ \Pi_0$. In this way, the final state of each target and auxiliary system is independent of $x_{ijk}$ and the algorithm in Fig.~\ref{figgeneral} can be applied directly.}\\

In Section~\ref{secsqrt}, we use similar ideas and present an algorithm that simulates the action of the $n$-switch for unitaries that satisfy pairwise commutation relations (and is therefore able to solve every FPP) with $O(n\sqrt{n})$ queries.

\section{Causal algorithms with query complexity $O(n\log{n})$} \label{seclog}
In this section, we present an algorithm that solves the FPP with the labeling used in Appendix~A of Ref.~\cite{1} with $O(n\log{n})$ queries.
First, we recall how the permutations are labeled and derive the pairwise commutation relations for this labeling.\\

The identity permutation is defined as:
\begin{align}
    \Pi_0:=U_{n-1}U_{n-2}...U_1U_0 \, .
\end{align}
The labeling of all other permutations $\Pi_x$ is based on the factorial number system; every $x$ is represented with $n-1$ integers $(a_{n-1},...,a_1)$ where $a_k\in \{0,1,...,k\}$: \footnote{This is Eq. (A4) in Ref.~\cite{1}.}
\begin{align}
    x=\sum_{k=1}^{n-1}a_k\cdot k! \, . \label{loglabel}
\end{align}
Starting with the identity permutation \mbox{$\Pi_0=U_{n-1}...U_1U_0$}, we obtain the permutation $\Pi_x$ by shifting first $U_1$ $a_1\in\{0,1\}$ steps to the right, then $U_2$ $a_2\in\{0,1,2\}$ steps to the right and so on. The labeling for $n=3$ is given as the first example in Subsection \ref{label3} (Eq. \eqref{fac31}-\eqref{fac3}). We call this labeling of the permutations the ``factoradic'' labeling.\\

Due to Proposition~\ref{proposition1}, we can read off the commutation relations for every pair of unitaries $U_j$ and $U_k$ (w.l.o.g. we assume here $j<k$). The two permutations we have to focus on are:
\begin{align}
    \Pi_{x^1_{jk}}&=U_{n-1}...U_{1}U_{0}U_jU_k \text{ and}\\
    \Pi_{x^2_{jk}}&=U_{n-1}...U_{1}U_{0}U_kU_j \, .
\end{align}
To construct the first permutation $\Pi_{x^1_{jk}}$ from the identity permutation $\Pi_0$, the unitary $U_j$ is first shifted $j$ steps to the right. In a second step, $U_k$ is shifted $k$ steps to the right, while all other unitaries are not shifted. Hence, the label of $\Pi_{x^1_{jk}}$ is $x^1_{jk}=k\cdot k!+ j\cdot j!$. To obtain the permutation $\Pi_{x^2_{jk}}$ from the identity permutation, the unitary $U_j$ is first shifted $j$ steps to the right and afterwards, $U_k$ is shifted $(k-1)$ steps to the right. The remaining unitaries are not shifted. Therefore, the label of the second permutation is $x^2_{jk}=(k-1)\cdot k!+ j\cdot j!$. If we combine this with Eq. \eqref{pair}, we obtain:
\begin{align}
\begin{split}
    U_jU_k=\omega^{(x^1_{jk}-x^2_{jk})\cdot y}\cdot U_kU_j=\omega^{k!\cdot y}\cdot U_kU_j\label{pairfac} \, .
\end{split}
\end{align}
Hence, whenever a unitary $U_k$ of the set for which the promise holds is commuted with a unitary of a smaller index, the result remains unchanged up to a phase~$\omega^{k!\cdot y}$. \footnote{This is equivalent to Eq. (A3) of Ref. \cite{1}. The difference is that here, we derive the pairwise commutation relations from the promise and in Ref.~\cite{1}, the pairwise commutation relations are used to show that unitaries which satisfy the promise exist (an explicit example of unitaries that satisfy the promise can also be found in Appendix~A of Ref.~\cite{1}).}\\

To introduce the idea, we give the algorithm for $n=4$ in the next subsection and generalize the procedure thereafter. Note that the query complexity of this example is actually worse than with the conventional method; it requires 18 queries, while the most efficient simulation of the quantum-$4$-switch calls twelve black-box unitaries.\footnote{Following the example given in Subsection~\ref{secexample37}, the shortest string containing all permutations of the four unitaries has twelve elements, for example $U_1U_2U_3U_4U_1U_2U_3U_1U_4U_2U_1U_3$ \cite{newey}.} Nonetheless, it is an instructive example whose generalization results in a significant reduction of the query complexity. As a further remark, we want to point out that we use the notation of controlled unknown unitaries merely for convenience. Note, however, that controlling unknown unitaries is impossible within the standard quantum circuit model \cite{Ara_jo_2014} but can be realized in the interferometric type of setups \cite{cunknown1, cunknown2, cunknown3}. At the end of this section (Subsection~\ref{controlofunk}), we show that it is possible to rewrite the algorithm in a form that does not control unknown unitaries.

\subsection{The algorithm for n=4}
For our purpose, it is useful to represent every number $x\in \{0,1,...,4!-1\}$ in a basis of bits $c^x_{k,i}\in\{0,1\}$. More precisely, we identify the state $\ket{x}_c$ with a six-qubit state
\begin{align}
    \ket{x}_c= \bigotimes_{\substack{1\leq k\leq 3\\1\leq i\leq 2}}  \ket{c^x_{k,i}} \, .
\end{align}
The bits $c^x_{k,i}$ satisfy the following equation:
\begin{align}
\begin{split}
    x=&c_{3,1}^x\cdot 12+c_{3,2}^x\cdot 6\\
    &+c_{2,1}^x\cdot 2+c_{2,2}^x\cdot 2\\
    &+c_{1,1}^x\cdot 1+c_{1,2}^x\cdot 1 \, . \label{ex4basis}
\end{split}
\end{align}
For example, $x=16$ can be written as $16=12+2+2$. Hence, $\ket{x=16}_c$ is represented by $\ket{c^{16}_{3,1}=1}\otimes\ket{c^{16}_{3,2}=0}\otimes\ket{c^{16}_{2,1}=1}\otimes\ket{c^{16}_{2,2}=1}\otimes\ket{c^{16}_{1,1}=0}\otimes\ket{c^{16}_{1,2}=0}$. It is simple to check that indeed every $x\in \{0,1,...,4!-1\}$ can be represented in this way. Note that this representation is not unique and most numbers can be decomposed in more than one way. For our purpose, it is enough to choose one representation for every $x\in \{0,1,...,4!-1\}$.\\

\begin{figure}[hbt!]
\smaller[2]
\begin{center}
$\input{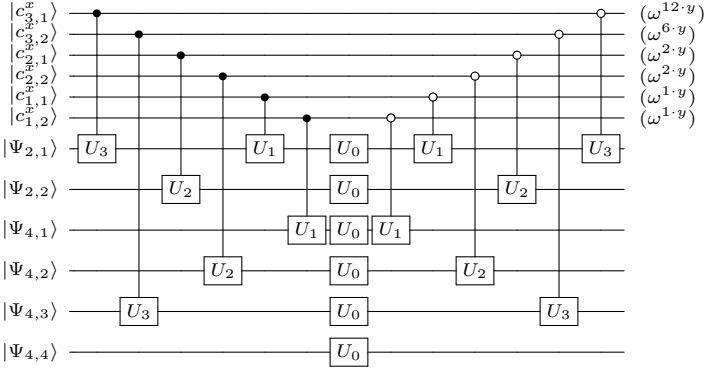}$
\end{center}
\normalsize
\caption{Implementation of the transformation $T_4^{\text{FPP}}$ for the factoradic labeling of the permutations.}
\label{figlog4}
\end{figure}

In accordance with the methods introduced in Subsection~\ref{newmethod}, we will show that the algorithm in Fig.~\ref{figlog4} can implement the transformation $T_4^{\text{FPP}}$ for this FPP. To see this, we look first at the target system $\ket{\Psi_{2,1}}$. The gates $U_3$, $U_1$ and $U_0$ act on this system but the order in which they are applied depends on the control qubits $\ket{c^x_{3,1}}$ and $\ket{c^x_{1,1}}$. If both control qubits are in the state $\ket{0}$, this target system ends up in the state $U_3U_1U_0\ket{\Psi_{2,1}}$. On the other hand, if one or both of the two control qubits are in the state $\ket{1}$, the order of these gates is different. Nevertheless, due to the pairwise commutation relations given in Eq.~\eqref{pairfac}, we can always rewrite the final state of the system $\ket{\Psi_{2,1}}$ into $U_3U_1U_0\ket{\Psi_{2,1}}$. By doing so, a phase is obtained whenever two unitaries are commuted (see Table~\ref{tablelog4}).\\

\begin{table}[h!]
\begin{tabular}{l|l|lllll}
$\ket{c^x_{3,1}}$ & $\ket{c^x_{1,1}}$ & \multicolumn{5}{c}{Final state of $\ket{\Psi_{2,1}}$} \\\hline 
$\ket{0}$       & $\ket{0}$       & $U_3U_1U_0\ket{\Psi_{2,1}}$       & $=$ &                           &                     & $U_3U_1U_0\ket{\Psi_{2,1}}$ \\
$\ket{0}$       & $\ket{1}$       & $U_3U_0U_1\ket{\Psi_{2,1}}$       & $=$ &                           & $\omega^{1\cdot y}$ & $U_3U_1U_0\ket{\Psi_{2,1}}$ \\
$\ket{1}$       & $\ket{0}$       & $U_0U_1U_3\ket{\Psi_{2,1}}$       & $=$ & $\omega^{12\cdot y}$      &                     & $U_3U_1U_0\ket{\Psi_{2,1}}$ \\
$\ket{1}$       & $\ket{1}$       & $U_1U_0U_3\ket{\Psi_{2,1}}$       & $=$ & $\omega^{12\cdot y}\cdot$ & $\omega^{1\cdot y}$ & $U_3U_1U_0\ket{\Psi_{2,1}}$
\end{tabular}
\caption{Final state of the target system $\ket{\Psi_{2,1}}$ for all combinations of the control qubits $\ket{c^x_{3,1}}$ and $\ket{c^x_{1,1}}$. }\label{tablelog4}
\end{table}

For instance, the qubit $\ket{c^x_{3,1}}$ controls whether the gate $U_3$ is applied before the two gates $U_0$ and $U_1$, or after these two gates. Whenever $\ket{c^x_{3,1}}=\ket{1}$, $U_3$ needs to be commuted with two unitaries of a smaller index (independent of the order of these two gates). In this way, a factor of $\omega^{3!\cdot y}$ is picked up twice, which multiplies together to $\omega^{12\cdot y}$. Independently of this, $\ket{c^x_{1,1}}$ controls the order of $U_1$ and $U_0$. If $\ket{c^x_{1,1}}=\ket{1}$, 
$U_1$ needs to be commuted with $U_0$ and a relative phase of $\omega^{1!\cdot y}$ is picked up.\\

Analogous arguments hold for the other systems; the control qubits control the order in which the involved gates are applied on the target system. Due to the pariwise commutation relations, the final state of each target system can always be rewritten in a form that is independent of the control system. By doing so, we pick up a phase $\omega^{(\cdot)\cdot y}$ (the expression in brackets at the end of each line of a control qubit in Fig. \ref{figlog4}), whenever a control qubit is in the state $\ket{1}$.
For example, if $x=16$, which is represented by $c^{16}_{3,1}=c^{16}_{2,1}=c^{16}_{2,2}=1$ and $c^{16}_{3,2}=c^{16}_{1,1}=c^{16}_{1,2}=0$, the circuit realizes the following transformation:
\begin{align}
\begin{split}
    &T_4^{\text{FPP}}\ket{c^{16}_{3,1}=1}\otimes \ket{c^{16}_{3,2}=0}\otimes \ket{c^{16}_{2,1}=1}\otimes \ket{c^{16}_{2,2}=1}\\ &\hspace{1.0 cm}\otimes\ket{c^{16}_{1,1}=0}\otimes \ket{c^{16}_{1,2}=0}\otimes \ket{\Psi_{2,1}}\otimes \ket{\Psi_{2,2}}\\ &\hspace{1.0 cm}\otimes\ket{\Psi_{4,1}}\otimes \ket{\Psi_{4,2}}\otimes \ket{\Psi_{4,3}}\otimes \ket{\Psi_{4,4}}\\\\
    &=\ket{c^{16}_{3,1}=1}\otimes \ket{c^{16}_{3,2}=0}\otimes \ket{c^{16}_{2,1}=1}\otimes \ket{c^{16}_{2,2}=1}\\ &\hspace{0.5 cm}\otimes\ket{c^{16}_{1,1}=0}\otimes \ket{c^{16}_{1,2}=0}\otimes U_1U_0U_3\ket{\Psi_{2,1}}\otimes U_0U_2\ket{\Psi_{2,2}}\\ &\hspace{0.5 cm}\otimes U_1U_0\ket{\Psi_{4,1}}\otimes U_0U_2\ket{\Psi_{4,2}}\otimes U_3U_0\ket{\Psi_{4,3}}\otimes \ket{\Psi_{4,4}}\\\\
    &=\omega^{16\cdot y}\ \ket{c^{16}_{3,1}=1}\otimes \ket{c^{16}_{3,2}=0}\otimes \ket{c^{16}_{2,1}=1}\otimes \ket{c^{16}_{2,2}=1}\\
    &\hspace{0.5 cm}\otimes\ket{c^{16}_{1,1}=0}\otimes \ket{c^{16}_{1,2}=0}\otimes U_3U_1U_0\ket{\Psi_{2,1}}\otimes U_2U_0\ket{\Psi_{2,2}}\\
    &\hspace{0.5 cm}\otimes U_1U_0\ket{\Psi_{4,1}}\otimes U_2U_0\ket{\Psi_{4,2}}\otimes U_3U_0\ket{\Psi_{4,3}}\otimes \ket{\Psi_{4,4}} \, .
\end{split}
\end{align}
Here, we obtain the factor of $\omega^{16\cdot y}$ as a composition of
\begin{align}
    \omega^{16\cdot y}=\omega^{12\cdot y}\cdot \omega^{2\cdot y}\cdot \omega^{2\cdot y} \, .
\end{align}
The first factor $\omega^{12\cdot y}$ stems from commuting $U_3$ with $U_1$ and $U_0$ on $\ket{\Psi_{2,1}}$, while the two factors of $\omega^{2\cdot y}$ arise from commuting $U_2$ with $U_0$ on $\ket{\Psi_{2,2}}$ and $\ket{\Psi_{4,2}}$, respectively. Since every \mbox{$x\in \{0,1,...,4!-1\}$} can be represented as in Eq.~\eqref{ex4basis}, the algorithm in Fig.~\ref{figlog4} applies the following transformation for every such $x$:
\begin{align}
\begin{split}
    &T_4^{\text{FPP}}\ket{x}_c\otimes\ket{\Psi_{2,1}}\otimes \ket{\Psi_{2,2}}\otimes\ket{\Psi_{4,1}}\\ 
    &\hspace{1.5 cm}\otimes \ket{\Psi_{4,2}}\otimes \ket{\Psi_{4,3}}\otimes \ket{\Psi_{4,4}}\\\\
    &=\omega^{x\cdot y}\ \ket{x}_c\otimes U_3U_1U_0\ket{\Psi_{2,1}}\otimes U_2U_0\ket{\Psi_{2,2}}\otimes U_1U_0\ket{\Psi_{4,1}}\\ &\hspace{1.5 cm}\otimes U_2U_0\ket{\Psi_{4,2}}\otimes U_3U_0\ket{\Psi_{4,3}}\otimes \ket{\Psi_{4,4}} \, .
\end{split}
\end{align}
Note, that the final state of the target system is independent of the control system $\ket{x}$. Hence, we can use our algorithm and the procedure introduced in Subsection~\ref{newmethod} to solve this specific FPP. More precisely, applying a Fourier transform, the control system is initialized in an equal superposition of all $\ket{x}$. After applying the algorithm, a measurement of the control system in the Fourier basis will yield the desired value of~$y$.\\

As mentioned before, the query complexity of this example is actually worse than with the simulation of the quantum-4-switch. Nevertheless, for larger $n$, the method that we introduce here solves this specific FPP with only $O(n\log{n})$ queries. The reason for this scaling advantage comes from the fact that every number $x\in\{0,1,..., n!-1\}$ can be represented with $O(n\log{n})$ bits $c^x_{k,i}$ (remember that $n!\leq 2^{n\log_2{n}}$). In the algorithms presented here, every such bit corresponds to a control qubit and only two queries couple to every control qubit.\\

As a remark, note that the target system $\ket{\Psi_{4,4}}$ is technically redundant. Moreover, the target system $\ket{\Psi_{4,1}}$ and the corresponding control qubit $\ket{c^x_{1,2}}$ are also not needed since the five bits $c_{3,1}^x$, $c_{3,2}^x$, $c_{2,1}^x$, $c_{2,2}^x$ and $c_{1,1}^x$ are sufficient to represent every number $x\in \{0,1,...,4!-1\}$ in the form given in Eq.~\eqref{ex4basis}. In this specific example for $n=4$, we kept them in for completeness. It turns out that for every $n$, there are some target systems that can be left out. In order to keep the notation as simple as possible, and since this does not affect the overall scaling, we refrain from doing so.

\begin{widetext}

\begin{figure}[h!]
\smaller[2]
\begin{subfigure}{0.95\textwidth}
\begin{center}
$\input{circbig}$\\
\subcaption{$C_k$ and $V_k$ ($\tilde{C}_k$ and $\tilde{V}_k$) are used as a shorthand notation and are defined in Fig.~\ref{subfigureb} (Fig.~\ref{subfigurec}).}
\end{center}
\end{subfigure}
\begin{subfigure}{0.45\textwidth}
\begin{flushright}
$\input{circV}$
\end{flushright}
\subcaption{Definition of $C_k$ and $V_k$\\ (The black dot denotes a control on $\ket{1}$, in symbols: $\ketbra{0}{0}\otimes \mathds{1}+\ketbra{1}{1}\otimes U_k$.)}\label{subfigureb}
\end{subfigure}\hfill
\begin{subfigure}{0.45\textwidth}
\begin{flushright}
$\input{circVtilde}$
\end{flushright}
\subcaption{Definition of $\tilde{C}_k$ and $\tilde{V}_k$\\(The white dot denotes a control on $\ket{0}$, in symbols: $\ketbra{0}{0}\otimes U_k+\ketbra{1}{1}\otimes \mathds{1}$.)}\label{subfigurec}
\end{subfigure}\\
\normalsize
    \caption{The quantum circuit implementing the transformation $T_n^{\text{FPP}}$ for the FPP with the factoradic labeling of the permutations: Depending on the state of the control qubits, the gates are applied on the target systems in a certain order. Due to the pairwise commutation relations, the final state of each target system can always be reordered but certain phases are picked up when two unitaries are commuted. For every $x\in \{0,1,...,n!-1\}$, these phases multiply together to $\omega^{x\cdot y}$ (see main text).
    }
    \label{figlogngen}
\end{figure}

\end{widetext}

\subsection{The algorithm for every $n$}

In this section, we show how the idea of the above example can be generalized to solve this specific FPP for arbitrary $n$. As above, it is convenient to introduce a specific representation of the control state $\ket{x}$ into qubits. More precisely, we use $(n-1)\cdot \ceil*{\log_2{n}}$ control qubits $\ket{c^x_{k,i}}$ ($c^x_{k,i}\in\{0,1\}$) where $k=1,2,...,n-1$ and $i=1,2,...,\ceil*{\log_2{n}}$. For convenience, we define $\hat{i}:= \ceil*{\log_2{n}}$. The state $\ket{x}_c$ is identified with
\begin{align}
    \ket{x}_c= \bigotimes_{\substack{1\leq k\leq n-1\\1\leq i\leq \ceil*{\log_2{n}}}}  \ket{c_{k,i}^x} \,  ,
\end{align}
where the bits $c_{k,i}^x$ satisfy the equation
\begin{align}
    x=\sum_{k=1}^{n-1}\sum_{i=1}^{\hat{i}}\ c_{k,i}^x\cdot \ceil*{\frac{k}{2^i}}\cdot k! \, .
    \label{equlog}
\end{align}
While the motivation for this basis will become clearer below, we give the formal proof that every $x\in\{0,1,...,n!-1\}$ can be represented in this way in Appendix \ref{appedixA}. \footnote{Note that the representation of $x$ in this basis is not unique. For our purpose, it is enough to pick one such representation for every $x$. Furthermore, this representation is related to the factorial number system that is used to label the permutations $\Pi_x$ in Eq.~\eqref{loglabel}. More precisely, $c_{k,i}^x\cdot \ceil*{\frac{k}{2^i}}$ is a representation of $a_k$ into $\ceil*{\log_2{n}}$ bits.}\\

Now, we will show that the quantum circuit given in Fig.~\ref{figlogngen} solves this specific FPP with $O(n\log{n})$ queries. The control system consists of the \mbox{$(n-1)\cdot \ceil*{\log_2{n}}$} control qubits $\ket{c^x_{k,i}}$ introduced above. Furthermore, we use several target systems $\ket{\Psi_{2^i,j}}$, where \mbox{$i\in \{1,2,...,\hat{i}:= \ceil*{\log_2{n}}\}$} and \mbox{$j\in \{1,2,3,...,2^i\}$}. As for every algorithm in this article, they are initialized in an arbitrary $d$-dimensional state.\\

The idea of the algorithm is, as usual, that the black-box gates act on the target systems in a certain order and due to the pairwise commutation relations certain phases are picked up by rewriting the final state of each target system. More precisely, we will show that whenever a control qubit $\ket{c^x_{k,i}}$ is in the state $\ket{1}$, we obtain a relative phase of $\omega^{\ceil*{\frac{k}{2^i}}\cdot k!\cdot y}$, independent of the states of the other control qubits. For every $x\in \{0,1,...,n!-1\}$ these phases multiply together to~$\omega^{x\cdot y}$:
\begin{widetext}

\begin{align}
\begin{split}
    T_n^{\text{FPP}}\ket{x}_c\otimes\left( \bigotimes_{\substack{1\leq i\leq \hat{i}\\ 1\leq j\leq 2^i}}  \ket{\Psi_{2^i,j}}\right)&=T_n^{\text{FPP}} \left(\bigotimes_{\substack{1\leq k\leq n-1\\1\leq i\leq \hat{i}}}  \ket{c_{k,i}^x}\right)\otimes\left( \bigotimes_{\substack{1\leq i\leq \hat{i}\\ 1\leq j\leq 2^i}}  \ket{\Psi_{2^i,j}}\right)\\
    &=\left(\prod_{\substack{1\leq k\leq n-1\\1\leq i\leq \hat{i}}}\omega^{c_{k,i}^x\cdot\ceil*{\frac{k}{2^i}}\cdot k!\cdot y }\right)\cdot \left(\bigotimes_{\substack{1\leq k\leq n-1\\1\leq i\leq \hat{i}}} \ket{c_{k,i}^x}\right) \otimes\left( \bigotimes_{\substack{1\leq i\leq \hat{i}\\ 1\leq j\leq 2^i}}  (...U_{2^i+j}U_jU_0)\ket{\Psi_{2^i,j}}\right)\\
    &= \omega^{x \cdot y}\ \ket{x}_c\otimes\left( \bigotimes_{\substack{1\leq i\leq \hat{i}\\ 1\leq j\leq 2^i}} (...U_{2^i+j}U_jU_0)\ket{\Psi_{2^i,j}}\right) \, . \label{translogn}
\end{split}
\end{align}
\end{widetext}
Here, we used the representation of $x$ in the basis given in Eq.~\eqref{equlog} to show that all the phases accumulate to:
\begin{align}
\begin{split}
     \prod_{\substack{1\leq k\leq n-1\\1\leq i\leq \hat{i}}}\omega^{c_{k,i}^x\cdot\ceil*{\frac{k}{2^i}}\cdot k!\cdot y }&=\omega^{\left(\sum\limits_{k=1}^{n-1}\sum\limits_{i=1}^{\hat{i}}\ c_{k,i}^x\cdot \ceil*{\frac{k}{2^i}}\cdot k!\cdot y \right)}=\omega^{x\cdot y} \, .
\end{split}
\end{align}
One can observe that the target system becomes independent of the control system. Hence, if the control system is initialized in an equal superposition of all $x$, the circuit applies, by linearity, the transformation
\begin{align}
\begin{split}
    &\left(\frac{1}{\sqrt{n!}}\sum_{x=0}^{n!-1}\ket{x}_c\right)\otimes\left( \bigotimes_{\substack{1\leq i\leq \hat{i}\\ 1\leq j\leq 2^i}}  \ket{\Psi_{2^i,j}}\right)\\
    &\mapsto\left(\frac{1}{\sqrt{n!}}\sum_{x=0}^{n!-1}\omega^{x \cdot y}\ket{x}_c\right)\otimes\left( \bigotimes_{\substack{1\leq i\leq \hat{i}\\ 1\leq j\leq 2^i}} (...U_{2^i+j}U_jU_0)\ket{\Psi_{2^i,j}}\right) \label{eqlogtrans}
\end{split}
\end{align}
and all target systems factorize out at the end of the algorithm. After applying the inverse Fourier transform to the control system, the correct value of $y$ can be read out with a measurement of the control system in the computational basis. For a better understanding of the algorithm, in addition to the example of $n=4$ in the last subsection, we give the circuit for $n=8$ in Appendix~\ref{appedixC}.

\subsubsection{How the algorithm works}
To show that this algorithm realizes the desired transformation in Eq.~\eqref{translogn}, we focus on one target system $\ket{\Psi_{2^i,j}}$ and all the gates that act on it. These are exactly the gates $U_k$ that satisfy $k\equiv j \Mod{2^i}$. The order in which these gates are applied on the target system $\ket{\Psi_{2^i,j}}$ depends on the states of the corresponding control qubits $\ket{c^x_{k,i}}$ (see Fig.~\ref{figtargetsys}).

\begin{figure}[h!]
\smaller[2]
\begin{flushright}
$\input{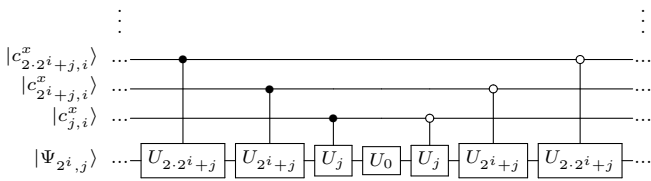}$
\end{flushright}
\normalsize
    \caption{Each control qubit $\ket{c^x_{k,i}}$ controls whether the gate $U_k$ (with $k\equiv j \Mod{2^i}$) is applied either before or after an entire block of matrices with a smaller index. For example, if $\ket{c^x_{2^i+j,i}=1}$ the gate $U_{2^i+j}$ is applied before $U_0$ and $U_j$, while if $\ket{c^x_{2^i+j,i}=0}$ it is applied after these two gates. Rewriting the final  state leads to a factor of $\omega^{2\cdot (2^i+j)!\cdot y}$ whenever $\ket{c^x_{2^i+j,i}=1}$.}
    \label{figtargetsys}
\end{figure}

On the one hand, if all control qubits are in the state $\ket{0}$, $U_0$ acts first, then $U_j$ and so on. This way, the final state of the target system becomes $(...U_{2\cdot 2^i+j}U_{2^i+j}U_jU_0)\ket{\Psi_{2^i,j}}$. On the other hand, if some of the control qubits are in the state $\ket{1}$, the gates act in a different order. The structure of the algorithm is chosen such that a gate is either applied immediately before or after an entire block of unitaries with a smaller index. By rewriting the final state into the form $(...U_{2\cdot 2^i+j}U_{2^i+j}U_jU_0)\ket{\Psi_{2^i,j}}$, a unitary has to be commuted either with all unitaries within this block or with none of them. One can check that every gate $U_k$ that appears on the target system has to be commuted with $\ceil*{\frac{k}{2^i}}$ unitaries of smaller index if and only if $\ket{c^x_{k,i}=1}$, leading to an additional factor of $\omega^{\ceil*{\frac{k}{2^i}}\cdot k!\cdot y}$. This shows that the algorithm realizes the transformation given in Eq.~\eqref{translogn}.\\


\subsubsection{Query complexity}
To count the total number of black-box unitaries that are used in this algorithm, we observe that for every $k\in \{1,2, ..., n-1\}$, the gate $U_k$ appears only in $V_k$ and $\tilde{V}_k$. In each of them, it is used exactly $\hat{i}=\ceil*{\log_2{n}}$ times. In addition, the gate $U_0$ acts on each target system exactly once. In total, there are
\begin{align}
    2+4+8+...+2^{\ceil*{\log_2{n}}}=2^{\ceil*{\log_2{n}}+1}-2
\end{align}
target systems. To see this, remember that the target systems are labeled with $\ket{\Psi_{2^i, j}}$, where $i=1,2,...,\ceil*{\log_2{n}}$ and $j=1,2,3,...,2^i$). Hence, the total number of black-box calls is
\begin{align}
\begin{split}
    Q&= 2\cdot (n-1)\cdot \ceil*{\log_2{n}} + 2^{\ceil*{\log_2{n}}+1}-2\\
    &<2\cdot (n-1)\cdot  (\log_2{n}+1) +2^{\log_2{n}+2}-2\\
    &=2\cdot (n-1)\cdot  (\log_2{n}+1) +4\cdot n-2\, .\\
\end{split}
\end{align}
We conclude that the query complexity of this algorithm is $O(n\log{n})$.

\subsubsection{Control of unknown unitaries}\label{controlofunk}
In this circuit, we control unknown unitaries. This operation is not well-defined within the standard quantum circuit model \cite{Ara_jo_2014}. Nevertheless, one can circumvent this issue by introducing auxiliary systems. More precisely, for every $k\in\{1,2,...,n-1\}$, we add an auxiliary system $\ket{a_k}$ initialized in an arbitrary $d$-dimensional state. Whenever an unknown unitary $U_k$, controlled on $\ket{c^x_{k,i}}$, shall be applied on $\ket{\Psi_{2^i, k \Mod{2^i}}}$, we perform instead a controlled swap of $\ket{\Psi_{2^i, k \Mod{2^i}}}$ and $\ket{a_k}$:

\smaller[1]
\begin{flushright}
$\Qcircuit @C=1em @R=1em {
\lstick{\ket{c^x_{k,  i}}}              & \ctrl{2}   & \qw &             &  &  &  &  &  &  &  &  &  & \lstick{\ket{c^x_{k,  i}}}              & \ctrl{7} & \qw        & \ctrl{7} & \qw \\
                                      &            &     &             &  &  &  &  &  &  &  &  &  &                                       &          &            &          &     \\
\lstick{\ket{\Psi_{2^i,k \Mod{2^i}}}} & \gate{U_k} & \qw & \Rightarrow &  &  &  &  &  &  &  &  &  & \lstick{\ket{\Psi_{2^i,k \Mod{2^i}}}} & \qswap   & \qw        & \qswap   & \qw \\
                                      &            &     &             &  &  &  &  &  &  &  &  &  &                                       &          &            &          &     \\
                                      &            &     &             &  &  &  &  &  &  &  &  &  & \lstick{\ket{a_{1}}}                  & \qw      & \qw        & \qw      & \qw \\
                                      &            &     &             &  &  &  &  &  &  &  &  &  &                        &          &        \vdots    &          &     \\
                                      &            &     &             &  &  &  &  &  &  &  &  &  &                                       &          &            &          &     \\
                                      &            &     &             &  &  &  &  &  &  &  &  &  & \lstick{\ket{a_{k}}}                  & \qswap   & \gate{U_k} & \qswap   & \qw \\
                                      &            &     &             &  &  &  &  &  &  &  &  &  &                       &          &         \vdots   &          &     \\
                                      &            &     &             &  &  &  &  &  &  &  &  &  &                                       &          &            &          &     \\
                                      &            &     &             &  &  &  &  &  &  &  &  &  & \lstick{\ket{a_{n-1}}}                & \qw      & \qw        & \qw      & \qw
}$
\end{flushright}
\normalsize

If the control qubit is in the state $\ket{c^x_{k,i}=1}$, the two systems are swapped and the gate $U_k$ is applied on the target system. On the other hand, if the control qubit is in the state $\ket{c^x_{k,i}=0}$, the two systems are not swapped and the gate $U_k$ is applied on the auxiliary system $\ket{a_k}$ instead. In the case where the gates are applied conditioned on $\ket{c^x_{k,i}=0}$ (for the boxes $\tilde{V}$), the swaps are also conditioned on $\ket{c^x_{k,i}=0}$ (all black dots are replaced by white dots).\\

It is important to ensure that this replacement does not affect the functionality of our algorithm. To see that this is true, note that for every $i\in\{1,2,...,\hat{i}\}$, the gate $U_k$ acts exactly once on the target system $\ket{\Psi_{2^i, k \Mod{2^i}}}$ and once on the auxiliary system $\ket{a_k}$, independent of the state of the control qubit $\ket{c^x_{k,i}}$. (The control qubit controls only if it is first applied on the target system and thereafter on the auxiliary system or vice versa.) In total, each auxiliary system ends up in the state $(U_k)^{\hat{i}}\ket{a_k}$, independent of the state of the control system. In this way, the auxiliary systems factorize out when the control system is initialized in a superposition of all states $x\in\{0,1,...,n!-1\}$ and do not affect the outcome of the measurement of the control system at the very end of the algorithm.

\subsection{$O(n\log{n})$ causal algorithms for every FPP?}
A natural question that appears is whether it is possible to solve other FPPs with $O(n\log{n})$ queries as well. The algorithms that we presented in this section can only be used for this specific labeling of the permutations, since we explicitly use the relations $U_jU_k=\omega^{k!\cdot y}\ U_kU_j$. If the permutations are labeled differently, the pairwise phases will change and the above algorithm cannot be used directly. Nevertheless, we think that the structure of our algorithm can be adopted to solve other FPPs as well. The idea that a certain phase $\omega^{\phi(k,i)\cdot y}$ is picked up whenever a control qubit $\ket{c^x_{k,i}}$ is in the state $\ket{1}$ (by using a structure as in Fig.~\ref{figtargetsys}) can be used for different pairwise commutation relations as well. If every $x\in\{0,1,..., n!-1\}$ can be written as
\begin{align}
    x=\sum_{k, i} c_{k,i}^x\cdot \phi(k,i) \label{xrep}
\end{align}
for some bits $c^x_{k,i}\in\{0,1\}$, we can use the control qubits $\ket{c^x_{k,i}}$ as the control system $\ket{x}_c$:
\begin{align}
    \ket{x}_c= \bigotimes_{k, i} \ket{c_{k,i}^x} \, .
\end{align}
By initializing the control system in an equal superposition of all $x\in\{0,1,..., n!-1\}$, such an algorithm will apply the transformation
\begin{align}
\begin{split}
    &\left(\frac{1}{\sqrt{n!}}\sum_{x=0}^{n!-1}\ket{x}_c\right)\otimes \bigotimes^{m}_{i=1}\ket{\Psi_i}\\
    &\mapsto \left(\frac{1}{\sqrt{n!}}\sum_{x=0}^{n!-1}\omega^{x\cdot y}\ket{x}_c\right)\otimes\bigotimes^{m}_{i=1}f^i_0(U_0,...,U_{n-1})\ket{\Psi_i} \, ,
    \end{split}
\end{align}
where for every $x$, the phase $\omega^{x\cdot y}$ is obtained as the product of:
\begin{align}
    \omega^{x\cdot y}=\prod_{k, i} \omega^{c_{k,i}^x\cdot \phi(k,i)\cdot y} \, .
\end{align}
Here, the equality follows from Eq.~\eqref{xrep}. Again, the solution $y$ can be read out after applying the inverse Fourier transform to the control system.\\

Since every number $x\in\{0,1,..., n!-1\}$ can in principle be represented with $O(n\log{n})$ bits ($n!\leq 2^{n\log_2{n}}$) and since two queries couple to every control qubit, it seems likely that every FPP can be solved with $O(n\log{n})$ queries. The crucial point is whether it is possible to find an implementation as in Fig.~\ref{figtargetsys}, a combination of gates and target systems such that this can be done efficiently. The disadvantage of this procedure is that it requires some rather involved combinatorics and that one has to adapt this algorithm by hand. While it remains open whether this is always possible, we present in the next section an algorithm that can solve \textit{every} FPP with $O(n\sqrt{n})$ queries, independent of the labeling of the permutations.

\section{A causal algorithm that solves every FPP with $O(n\sqrt{n})$ queries} \label{secsqrt}

In this section, we present an algorithm that solves every Fourier promise problem with $O(n\sqrt{n})$ queries.
The main idea is based on the fact that the existence of pairwise commutation relations \mbox{($U_jU_k=\alpha_{jk}\ U_kU_j$)} allows us to rewrite every permutation \mbox{$\Pi_x=U_{\sigma_x(n-1)}...U_{\sigma_x(1)}U_{\sigma_x(0)}$} into:
\begin{align}
    U_{\sigma_x(n-1)}...U_{\sigma_x(1)}U_{\sigma_x(0)}=\alpha_x\ U_{n-1}...U_{1}U_{0} \, ,
\end{align}
where the total phase $\alpha_x$ is a product of pairwise phases $\alpha_{jk}$. We use this fact to decompose the total phase of every permutation into different factors and simulate each factor on a different target system. So instead of simulating every permutation $\Pi_x$ on its own (which requires the simulation of the (full) $n$-switch and hence $O(n^2)$ queries), we construct other expressions that can simulate these factors and call only $O(n\sqrt{n})$ gates in total. Via the ``phase-kickback,'' all these factors accumulate in the control system and multiply together to the total phase~$\alpha_x$.\\

More precisely, we decompose every permutation \mbox{$\Pi_x=U_{\sigma_x(n-1)}...U_{\sigma_x(1)}U_{\sigma_x(0)}$} into blocks of length \mbox{$\hat{n}:=\ceil*{\sqrt{n}}$} (the last block contains all remaining unitaries and is usually shorter). The number of blocks obtained in this way is \mbox{$\hat{k}:=\ceil*{\frac{n}{\hat{n}}}$}. Formally, we define:

\begin{widetext}
\begin{definition}
For every permutation \mbox{$\Pi_x=(U_{\sigma_x(n-1)}...U_{\sigma_x((\hat{k}-1)\cdot \hat{n})})...\ 
    (U_{\sigma_x(2\cdot \hat{n}-1)}...U_{\sigma_x(\hat{n})})\ 
    (U_{\sigma_x(\hat{n}-1)}...U_{\sigma_x(0)})$} (the organization into blocks is merely for illustrative reasons) of the $n$ unitaries, let
\begin{align}
\Pi_{xk}&:=
\begin{cases}
    [U_{\sigma_x(n-1)}...U_{\sigma_x(\hat{n})}]\ 
    (U_{\sigma_x(\hat{n}-1)}...U_{\sigma_x(0)})\\
    [U_{\sigma_x(n-1)}...U_{\sigma_x((k+1)\cdot \hat{n})}]\ 
    (U_{\sigma_x((k+1)\cdot \hat{n}-1)}...U_{\sigma_x(k\cdot \hat{n})})\ 
    [U_{\sigma_x(k\cdot \hat{n}-1)}...U_{\sigma_x(0)}]\\
    (U_{\sigma_x(n-1)}...U_{\sigma_x((\hat{k}-1)\cdot \hat{n})})\ 
    [U_{\sigma_x((\hat{k}-1)\cdot \hat{n}-1)}...U_{\sigma_x(0)}]
\end{cases}
&
\begin{array}{l}
    k=0\\
    k=1,...,\hat{k}-2\\
    k=\hat{k}-1
\end{array} \label{pixk}
\\
\tilde{\Pi}^r_{xk}&:=
\begin{cases}
    \{U_{\sigma_x(k\cdot \hat{n}-1)}...U_{\sigma_x(0)}\}\ 
    \{U_{\sigma_x(n-1)}...U_{\sigma_x(k\cdot \hat{n})}\}\\
    \{U_{\sigma_x((\hat{k}-1)\cdot \hat{n}-1)}...U_{\sigma_x(0)}\}\ 
    \{U_{\sigma_x(n-1)}...U_{\sigma_x((\hat{k}-1)\cdot \hat{n})}\}
\end{cases}
&
\begin{array}{l}
    k=1,...,\hat{k}-2\\
    k=\hat{k}-1
\end{array} \label{pixktilde}
\end{align}
where $[U_{i_1}U_{i_2}...U_{i_j}]$ is defined to be the descending ordering of the unitaries $U_{i_1}U_{i_2}...U_{i_j}$, while $\{U_{i_1}U_{i_2}...U_{i_j}\}$ is the ascending ordering of them and $(U_{i_1}U_{i_2}...U_{i_j})=U_{i_1}U_{i_2}...U_{i_j}$ leaves the string invariant.

\end{definition}
\end{widetext}

As an example, consider \mbox{$n=9$} and \mbox{$\Pi_x=(U_3U_5U_8)\ (U_0U_2U_7)\ (U_4U_6U_1)$}:
\begin{align}
\begin{split}
    \Pi_{x0}&=(U_8U_7U_5U_3U_2U_0)\ (U_4U_6U_1) \, ,\\
    \Pi_{x1}&=(U_8U_5U_3)\ (U_0U_2U_7)\ (U_6U_4U_1) \, ,\\
    \Pi_{x2}&=(U_3U_5U_8)\ (U_7U_6U_4U_2U_1U_0) \, ,\\\\    \tilde{\Pi}^r_{x1}&=(U_1U_4U_6)\ (U_0U_2U_3U_5U_7U_8) \, \text{ and}\\
    \tilde{\Pi}^r_{x2}&=(U_0U_1U_2U_4U_6U_7)\ (U_3U_5U_8) \, .
\end{split}
\end{align}
They are defined in a way that $\Pi_{xk}$ and $\tilde{\Pi}^r_{xk}$ simulate exactly all pairwise phases $\alpha_{jk}$ for the unitaries within the block $(U_{\sigma_x((k+1)\cdot \hat{n}-1)}...U_{\sigma_x(k\cdot \hat{n})})$. If they act on different target systems $\ket{\Psi_k}$ and $\ket{\Phi_k}$ respectively, all the pairwise phases are accumulated and we obtain as the product the total phase of the original permutation $\Pi_x$:
\begin{lemma}\label{lemmasqrt}
For every set of ($d$-dimensional) unitaries $\{U_i\}_0^{n-1}$ that satisfy pairwise commutation relations, the following relation holds:
\begin{align}
\begin{split}
    &\bigotimes^{\hat{k}-1}_{k=0}\Pi_{xk}\ket{\Psi_k}\otimes \bigotimes^{\hat{k}-1}_{k=1}\tilde{\Pi}^r_{xk}\ket{\Phi_k}\\
    &\hspace{1.5cm}=\Pi_x\ket{\Psi_0}\otimes \bigotimes^{\hat{k}-1}_{k=1}\Pi\ket{\Psi_k}\otimes \Pi^r\ket{\Phi_k} \, .
    \label{toproofsqrtmain}
\end{split}
\end{align}
Here, $\Pi:=U_{n-1}...U_1U_0$ denotes the descending order of all unitaries, $\Pi^r:=U_{0}U_1...U_{n-1}$ denotes the ascending order of all unitaries and $\ket{\Psi_k},\ket{\Phi_k}\in \mathcal{H}^d$ are arbitrary $d$-dimensional states.
\end{lemma}
\begin{proof}
See Appendix~\ref{appedixB}.
\end{proof}
It turns out that the permutations $\Pi_{xk}$ and $\tilde{\Pi}^r_{xk}$, due to the fact that a large part of each of them is already ordered, can be simulated with a causal algorithm and $O(n\sqrt{n})$ queries. The algorithm that achieves this is presented in the next subsection.\\

In the algorithm, $(2\cdot\hat{k}-1)$ target systems, denoted as $\ket{\Psi_k}$ and $\ket{\Phi_k}$, as well as $n$ auxiliary systems $\ket{a_i}$ are used. All of them are initialized in an arbitrary $d$-dimensional state. The control system is a system of at least $n!$ dimensions and the algorithm applies, depending on the state $\ket{x}$ of the control system, the permutations $\Pi_{xk}$ on $\ket{\Psi_k}$ and the permutations $\tilde{\Pi}^r_{xk}$ on $\ket{\Phi_k}$. All the remaining gates $U_i$ act on the corresponding auxiliary system $\ket{a_i}$. In this way, the algorithm realizes the transformation
\begin{align}
\begin{split}
    &T_n^{\text{FPP}}\ket{x}_c\otimes\bigotimes^{\hat{k}-1}_{k=0}\ket{\Psi_k}\otimes \bigotimes^{\hat{k}-1}_{k=1}\ket{\Phi_k}\otimes\bigotimes^{n-1}_{i=0}\ket{a_i}\\
    &=\ket{x}_c\otimes\bigotimes^{\hat{k}-1}_{k=0}\Pi_{xk}\ket{\Psi_k}\otimes \bigotimes^{\hat{k}-1}_{k=1}\tilde{\Pi}^r_{xk}\ket{\Phi_k}\otimes\bigotimes^{n-1}_{i=0}(U_i)^{k_i}\ket{a_i}\\
    &=\ket{x}_c\otimes\Pi_x\ket{\Psi_0}\otimes \bigotimes^{\hat{k}-1}_{k=1}\Pi\ket{\Psi_k}\otimes \Pi^r\ket{\Phi_k}\otimes\bigotimes^{n-1}_{i=0}(U_i)^{k_i}\ket{a_i} \, .    \label{transsqrt}
\end{split}
\end{align}
Here, $k_i=\hat{n}+2\cdot \hat{k}-3$ is a constant that only depends on $n$ and Lemma~\ref{lemmasqrt} is used to rewrite the state in the second step. Except for the first target system $\ket{\Psi_0}$, which ends up in the state $\Pi_x\ket{\Psi_0}$, the final state of each target and auxiliary system is independent of $x$ and we conclude that this algorithm simulates the action of the $n$-switch for unitaries satisfying pairwise commutation relations and is therefore able to solve every Fourier promise problem. More precisely, as described in Subsection~\ref{newmethod}, with a Fourier transform, the control system is initialized in an equal superposition of all states $x\in\{0,1,...,n!-1\}$ and after the algorithm is applied, the solution~$y$ can be read out with a measurement in the Fourier basis.\footnote{For the comparison with Eq.~\eqref{deftfpp}, note that by using the promise $\Pi_{x}=\omega^{x\cdot y}\ \Pi_0$, the term $\ket{x}_c\otimes\Pi_x\ket{\Psi_0}$ in Eq.~\eqref{transsqrt} can be rewritten into $\omega^{x\cdot y}\ \ket{x}_c\otimes\Pi_0\ket{\Psi_0}$.} Note that at no point we refer to a specific labeling of the permutations $\Pi_x$ and hence, we can solve every possible Fourier promise problem with this causal algorithm. The advantage stems merely from the fact that every set of unitaries that satisfies the promise also satisfies pairwise commutation relations which imply that Lemma~\ref{lemmasqrt} holds.

\newcommand{\namepart}{X_\text{Part 2}}
\newcommand{\nzeilen}{16}

\begin{widetext}
    \begin{flushright}
    \begin{figure}[h!]
    \smaller[1]
    $\input{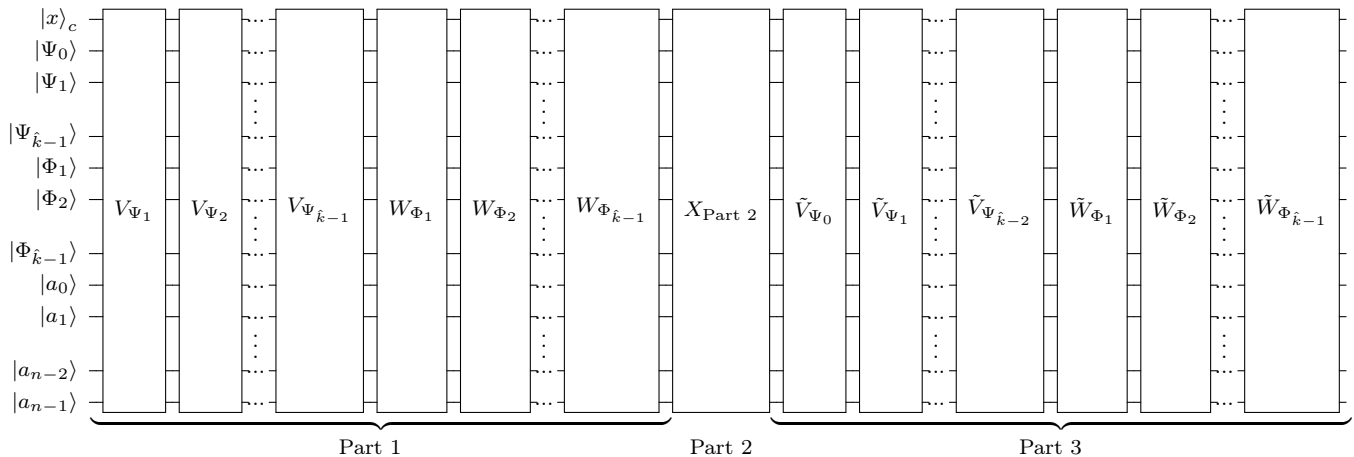}$
    \normalsize
    \caption{The quantum algorithm that implements $T_n^{\text{FPP}}$ for every FPP with $O(n\sqrt{n})$ queries. The circuit is decomposed into three parts that are explained below.}
    \label{figcircuitsqrt}
\end{figure}
\end{flushright}
\end{widetext} 

\subsection{The quantum algorithm}
Here we present the quantum circuit that realizes the transformation described in the last subsection (Eq.~\eqref{transsqrt}) and show that this algorithm uses $O(n\sqrt{n})$ queries. To keep the procedure as clear as possible, we divide the quantum circuit into three parts.

\subsubsection{Part 1}
First, all target systems $\ket{\Psi_k}$ undergo the transformations
\begin{align}
    \forall \ 1\leq k\leq \hat{k}-1:\ \ket{\Psi_k} \mapsto [U_{\sigma_x(k\cdot \hat{n}-1)}...U_{\sigma_x(0)}]\ket{\Psi_k} \, .
\end{align}
For each $\ket{\Psi_k}$, this is realized by the algorithm $V_{\Psi_k}$ given in Fig.~\ref{figpart1}. Here, depending on the state $\ket{x}$, in each step $i=0,1,...,n-1$, the target system $\ket{\Psi_k}$ is swapped with $\ket{a_i}$ if and only if $U_i$ is contained in $[U_{\sigma_x(k\cdot \hat{n}-1)}...U_{\sigma_x(0)}]$.

\newcommand{\vsmaller}{0}
\newcommand{\vhspace}{0.3}
\newcommand{\vC}{0.5}
\newcommand{\vR}{1}

\begin{figure}[h!]
\smaller[\vsmaller]
\begin{flushright}
$\input{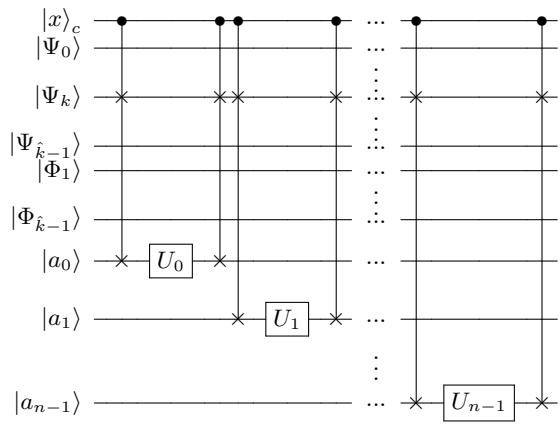}$
\end{flushright}
\normalsize
    \caption{Implementation of $V_{\Psi_k}$}
    \label{figpart1}
\end{figure}

To understand why this circuit realizes the above transformation, note that $[U_{\sigma_x(k\cdot \hat{n}-1)}...U_{\sigma_x(0)}]$ is by construction a block of $k\cdot \hat{n}$ unitaries in descending order and the unitary with the smallest index has to be applied first. By going step by step through each of the $n$ possible unitaries $U_0$ until $U_{n-1}$, exactly those unitaries contained in the ordered block are applied on the target system $\ket{\Psi_k}$ and all the others are applied on the corresponding auxiliary system. Each of these unitaries $V_{\Psi_k}$ consumes $n$ queries.\\

Similarly, the transformations
\begin{align}
    \forall \ 1\leq k\leq \hat{k}-1:\ \ket{\Phi_k} \mapsto \{U_{\sigma_x(n-1)}...U_{\sigma_x(k\cdot \hat{n})}\}\ket{\Phi_k}
\end{align}
are realized with the algorithm $W_{\Phi_k}$ given by the circuit in Fig.~\ref{figpart12}. Here, in each step, the target system $\ket{\Phi_k}$ is swapped with $\ket{a_i}$ if and only if the gate $U_i$ is contained in $\{U_{\sigma_x(n-1)}...U_{\sigma_x(k\cdot \hat{n})}\}$.

\begin{figure}[h!]
\smaller[\vsmaller]
\begin{flushright}
$\input{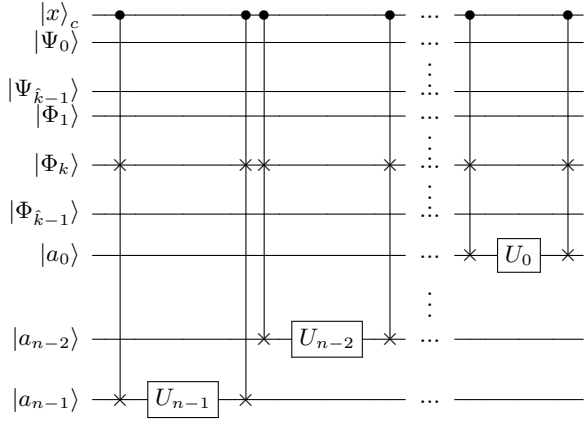}$
\end{flushright}
\normalsize
    \caption{Implementation of $W_{\Phi_k}$}
    \label{figpart12}
\end{figure}

The only difference is that the unitaries in this block are arranged in ascending order and the unitary with the highest index is applied first. Each of these $W_{\Phi_k}$ consumes again $n$ queries.

\subsubsection{Part 2}
In the second part, we realize the transformations:
\begin{align}
\begin{split}
    &\forall \ 0\leq k\leq \hat{k}-2:\\
    &[U_{\sigma_x(k\cdot \hat{n}-1)}...U_{\sigma_x(0)}] \ket{\Psi_k}\\
    & \mapsto(U_{\sigma_x((k+1)\cdot \hat{n}-1)}...U_{\sigma_x(k\cdot \hat{n})})[U_{\sigma_x(k\cdot \hat{n}-1)}...U_{\sigma_x(0)}] \ket{\Psi_k}\\\\
    &\text{and for} \ k= \hat{k}-1:\\
    &[U_{\sigma_x((\hat{k}-1)\cdot \hat{n}-1)}...U_{\sigma_x(0)}] \ket{\Psi_{\hat{k}-1}}\\
    & \mapsto(U_{\sigma_x(n-1)}...U_{\sigma_x((\hat{k}-1)\cdot \hat{n})})[U_{\sigma_x((\hat{k}-1)\cdot \hat{n}-1)}...U_{\sigma_x(0)}] \ket{\Psi_{\hat{k}-1}}
\end{split}
\end{align}
with the algorithm in Fig.~\ref{figpart2}. In each step \mbox{$i=0,1,...,\hat{n}-1$}, every $\ket{\Psi_k}$ is swapped with the auxiliary system $\ket{a_{\sigma_x(k\cdot \hat{n}+i)}}$. In this way, $U_{\sigma_x(k\cdot \hat{n}+i)}$ acts on $\ket{\Psi_k}$, and afterwards $\ket{\Psi_k}$ and $\ket{a_{\sigma_x(k\cdot \hat{n}+i)}}$ are swapped back. After $\hat{n}$ steps, the entire block $(U_{\sigma_x((k+1)\cdot \hat{n}-1)}...U_{\sigma_x(k\cdot \hat{n})})$ is applied on each target system $\ket{\Psi_k}$. (For $k=\hat{k}-1$, we already stop after the step in which $U_{\sigma_x(n-1)}$ is applied on $\ket{\Psi_{\hat{k}-1}}$.)

\begin{figure}[h!]
\smaller[2]
\begin{flushright}
$\input{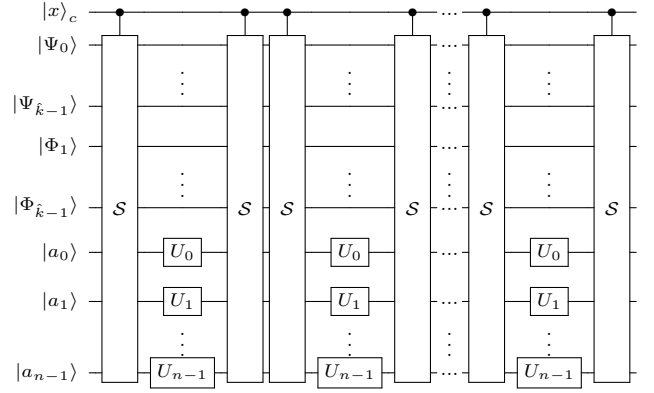}$
\end{flushright}
\normalsize
    \caption{Implementation of $X_{\text{Part 2}}$}
    \label{figpart2}
\end{figure}

Note that the target systems $\ket{\Phi_k}$ are unaffected by this part of the algorithm. In each of the $\hat{n}$ steps, $n$ queries are consumed. This part of the algorithm is similar to the algorithm presented in Subsection~\ref{secsimswitch}. The difference is that we swap several target systems simultaneously, instead of only one.

\subsubsection{Part 3}
Finally, the remaining blocks $[U_{\sigma_x(n-1)}...U_{\sigma_x((k+1)\cdot \hat{n})}]$ need to be applied on the target systems $\ket{\Psi_k}$:
\begin{align}
\begin{split}
    &\forall \ 0\leq k\leq \hat{k}-2:\\
    &(U_{\sigma_x((k+1)\cdot \hat{n}-1)}...U_{\sigma_x(k\cdot \hat{n})})[U_{\sigma_x(k\cdot \hat{n}-1)}...U_{\sigma_x(0)}] \ket{\Psi_k}\\
    & \hspace{2cm}\mapsto \Pi_{xk}\ket{\Psi_k} \, .
\end{split}
\end{align}

Similarly, the blocks $\{U_{\sigma_x(k\cdot \hat{n}-1)}...U_{\sigma_x(0)}\}$ need to be applied on $\ket{\Phi_k}$:
\begin{align}
\begin{split}
    &\forall \ 1\leq k\leq \hat{k}-1:\\
    &\{U_{\sigma_x(n-1)}...U_{\sigma_x(k\cdot \hat{n})}\}\ket{\Phi_k}\mapsto \tilde{\Pi}^r_{xk}\ket{\Phi_k} \, .
\end{split}
\end{align}
Both transformations are completely analogous to the procedure in Part 1 and require $2\cdot(\hat{k}-1)\cdot n$ queries in total.\\

\subsubsection{Query complexity}
The number of queries in Part 1 amounts to \mbox{$2\cdot (\hat{k}-1)\cdot n$}. For Part 2, we need $\hat{n}\cdot n$ queries, while for the last Part, \mbox{$2\cdot (\hat{k}-1)\cdot n$} black-box gates are called. Summing these together gives
\begin{align}
    Q=(\hat{n}+4\cdot \hat{k}-4)\cdot n
\end{align}
queries in total. Using \mbox{$\hat{n}:=\ceil*{\sqrt{n}}<\sqrt{n}+1$} and \mbox{$\hat{k}:=\ceil*{\frac{n}{\hat{n}}}<\frac{n}{\hat{n}}+1\leq \frac{n}{\sqrt{n}}+1=\sqrt{n}+1$} we obtain:
\begin{align}
    Q<(5\cdot \sqrt{n}+1)\cdot n \, .
\end{align}

It is important to ensure that the auxiliary systems factorize out at the end of the algorithm. To see that this is indeed true, we observe that in each part of the algorithm a gate $U_i$ acts either on a target system $\ket{\Psi_k}$, a target system $\ket{\Phi_k}$ or the auxiliary system $\ket{a_i}$. All together, for every $i\in\{0,1,...,n-1\}$, the gate $U_i$ appears exactly $\hat{n}+4\cdot \hat{k}-4$ times in the algorithm. Furthermore, it acts, independently of the permutation $\Pi_x$, on each of the $\hat{k}$ target systems $\ket{\Psi_k}$ and each of the $\hat{k}-1$ target systems $\ket{\Phi_k}$ exactly once. This is true since the expressions $\Pi_{xk}$ and $\tilde{\Pi}_{xk}$ are by themselves permutations of the $n$ unitaries and contain each $U_i$ exactly once. In all other remaining instances, $U_i$ acts on the auxiliary system $\ket{a_i}$, which therefore ends up in the state $(U_i)^{k_i}\ket{a_i}$ with \mbox{$k_i=\hat{n}+2\cdot \hat{k}-3$}, independent of the state of the control system $\ket{x}$. In total, we have shown that the algorithm realizes the desired transformation given in Eq.~\eqref{transsqrt} for every \mbox{$x\in\{0,1,...,n!-1\}$}:
\begin{align}
\begin{split}
    &T_n^{\text{FPP}}\ket{x}_c\otimes\bigotimes^{\hat{k}-1}_{k=0}\ket{\Psi_k}\otimes \bigotimes^{\hat{k}-1}_{k=1}\ket{\Phi_k}\otimes\bigotimes^{n-1}_{i=0}\ket{a_i}\\
    &=\ket{x}_c\otimes\bigotimes^{\hat{k}-1}_{k=0}\Pi_{xk}\ket{\Psi_k}\otimes \bigotimes^{\hat{k}-1}_{k=1}\tilde{\Pi}^r_{xk}\ket{\Phi_k}\otimes\bigotimes^{n-1}_{i=0}(U_i)^{k_i}\ket{a_i}
    \, .    
\end{split}
\end{align}
Hence, we conclude that this algorithm solves every Fourier promise problem with $O(n\sqrt{n})$ queries.

\section{Conclusion}
The introduction of indefinite causal structures raised the question of the existence of computational tasks which can be solved more efficiently using these structures, compared to causally ordered protocols.
Fourier promise problems were initially introduced to demonstrate that such a computational advantage exists even in the asymptotic case. The problems were shown to be solved with $n$ queries using the quantum-$n$-switch and it was expected that the most efficient solution with a causal protocol requires the simulation of the quantum-$n$-switch and, hence, $O(n^2)$ queries.
We showed that for the specific task of solving Fourier promise problems, the advantage of using the quantum-$n$-switch is significantly smaller than previously expected; in fact, we presented a causal quantum algorithm, within the standard quantum circuit model, which solves the same computational tasks almost as efficiently. More precisely, we presented a causal algorithm that solves a specific Fourier promise problem with $O(n\log{n})$ queries and conjectured that all problems of this class can be solved with a similar efficiency. Furthermore, we presented a causal algorithm that solves every Fourier promise problem with $O(n\sqrt{n})$ queries.\\

We conclude that for the specific class of problems considered here, the advantage of algorithms that use a quantum-controlled ordering of gates, compared to causally ordered algorithms, is smaller than first expected. Nevertheless, although we could show that the simulation of the quantum-$n$-switch is not the most efficient causal algorithm for solving FPPs, it is in principle possible to construct tasks which profit more from using the quantum-$n$-switch. One promising class of problems is already introduced in Ref.~\cite{taddei2020experimental}. The so-called Hadamard promise problems are a variation of the Fourier promise problems with the advantage that they are better suited for experimental realization. In Ref.~\cite{taddei2020experimental}, the authors present only one specific task with four unitaries without providing a generalization to an arbitrary number of gates. In this sense, it needs further investigation whether there is a significant asymptotic scaling advantage for Hadamard promise problems or whether the methods we developed in this work also apply to these problems. All in all, this raises the important challenge of finding computational tasks for which the quantum-$n$-switch and indefinite causal structures in general provide a significant advantage.





\section*{Acknowledgements}
We thank Mateus Araújo, Philippe Allard Guérin, Ämin Baumeler and Anne-Catherine de la Hamette for insightful discussions and useful comments on the manuscript. We  acknowledge  financial  support  from the  Austrian  Science  Fund  (FWF)  through  BeyondC (F7103-N38), the project no. I-2906, as well as support by  the  John Templeton Foundation through grant 61466, The Quantum Information Structure of Spacetime (qiss.fr), the Foundational Questions Institute (FQXi) and the research platform TURIS. The opinions expressed in this publication are those of the authors and do not necessarily reflect the views of the John Templeton Foundation. Furthermore, we are thankful for the very helpful tutorial on Q-circuit \cite{eastin2004qcircuit} that helped a lot to create the quantum circuits in \LaTeX.

\bibliography{bib}{}

\begin{thebibliography}{10}

\bibitem{hardy2005probability}
L.~Hardy, ``Probability theories with dynamic causal structure: A new framework
  for quantum gravity,'' {\em arXiv:gr-qc/0509120}, Sep 2005.

\bibitem{Zych_2019}
M.~Zych, F.~Costa, I.~Pikovski, and {\v{C}}.~Brukner, ``Bell’s theorem for
  temporal order,'' {\em Nature Communications}, vol.~10, Aug 2019.

\bibitem{Chiribella_2013}
G.~Chiribella, G.~M. D’Ariano, P.~Perinotti, and B.~Valiron, ``Quantum
  computations without definite causal structure,'' {\em Physical Review A},
  vol.~88, Aug 2013.

\bibitem{Oreshkov_2012}
O.~Oreshkov, F.~Costa, and {\v{C}}.~Brukner, ``Quantum correlations with no
  causal order,'' {\em Nature Communications}, vol.~3, Jan 2012.

\bibitem{Chiribella_2012}
G.~Chiribella, ``Perfect discrimination of no-signalling channels via quantum
  superposition of causal structures,'' {\em Physical Review A}, vol.~86, Oct
  2012.

\bibitem{Guerin_2016}
P.~A. Guérin, A.~Feix, M.~Araújo, and {\v{C}}.~Brukner, ``Exponential
  communication complexity advantage from quantum superposition of the
  direction of communication,'' {\em Physical Review Letters}, vol.~117, Sep
  2016.

\bibitem{PhysRevLett.120.120502}
D.~Ebler, S.~Salek, and G.~Chiribella, ``Enhanced communication with the
  assistance of indefinite causal order,'' {\em Physical Review Letters},
  vol.~120, p.~120502, Mar 2018.

\bibitem{salek2018quantum}
S.~Salek, D.~Ebler, and G.~Chiribella, ``Quantum communication in a
  superposition of causal orders,'' {\em arXiv:1809.06655}, Sep 2018.

\bibitem{chiribella2018indefinite}
G.~Chiribella, M.~Banik, S.~S. Bhattacharya, T.~Guha, M.~Alimuddin, A.~Roy,
  S.~Saha, S.~Agrawal, and G.~Kar, ``Indefinite causal order enables perfect
  quantum communication with zero capacity channel,'' {\em arXiv:1810.10457},
  Oct 2018.

\bibitem{Guo2020}
Y.~Guo, X.-M. Hu, Z.-B. Hou, H.~Cao, J.-M. Cui, B.-H. Liu, Y.-F. Huang, C.-F.
  Li, G.-C. Guo, and G.~Chiribella, ``Experimental transmission of quantum
  information using a superposition of causal orders,'' {\em Physical Review
  Letters}, vol.~124, p.~030502, Jan 2020.

\bibitem{goswami2020}
K.~Goswami, Y.~Cao, G.~A. Paz-Silva, J.~Romero, and A.~G. White, ``Increasing
  communication capacity via superposition of order,'' {\em Physical Review
  Research}, vol.~2, p.~033292, Aug 2020.

\bibitem{abbott2018communication}
A.~A. Abbott, J.~Wechs, D.~Horsman, M.~Mhalla, and C.~Branciard,
  ``Communication through coherent control of quantum channels,'' {\em
  Quantum}, vol.~4, p.~333, Sep 2020.

\bibitem{Gu_rin_2019}
P.~A. Guérin, G.~Rubino, and {\v{C}}.~Brukner, ``Communication through
  quantum-controlled noise,'' {\em Physical Review A}, vol.~99, Jun 2019.

\bibitem{Rubino_2021}
G.~Rubino, L.~A. Rozema, D.~Ebler, H.~Kristj\'ansson, S.~Salek,
  P.~Allard~Gu\'erin, A.~A. Abbott, C.~Branciard, {\v{C}}.~Brukner,
  G.~Chiribella, and P.~Walther, ``Experimental quantum communication
  enhancement by superposing trajectories,'' {\em Physical Review Research},
  vol.~3, p.~013093, Jan 2021.

\bibitem{Ara_jo_2017}
M.~Araújo, P.~A. Guérin, and {\"A}.~Baumeler, ``Quantum computation with
  indefinite causal structures,'' {\em Physical Review A}, vol.~96, Nov 2017.

\bibitem{Baumeler_2018}
{\"A}.~Baumeler and S.~Wolf, ``Computational tameness of classical non-causal
  models,'' {\em Proceedings of the Royal Society A: Mathematical, Physical and
  Engineering Sciences}, vol.~474, p.~20170698, Jan 2018.

\bibitem{Procopio_2015}
L.~M. Procopio, A.~Moqanaki, M.~Araújo, F.~Costa, I.~Alonso~Calafell, E.~G.
  Dowd, D.~R. Hamel, L.~A. Rozema, {\v{C}}.~Brukner, and P.~Walther,
  ``Experimental superposition of orders of quantum gates,'' {\em Nature
  Communications}, vol.~6, Aug 2015.

\bibitem{Rubino_2017}
G.~Rubino, L.~A. Rozema, A.~Feix, M.~Araújo, J.~M. Zeuner, L.~M. Procopio,
  {\v{C}}.~Brukner, and P.~Walther, ``Experimental verification of an
  indefinite causal order,'' {\em Science Advances}, vol.~3, p.~e1602589, Mar
  2017.

\bibitem{rubino2017experimental}
G.~Rubino, L.~A. Rozema, F.~Massa, M.~Araújo, M.~Zych, {\v{C}}.~Brukner, and
  P.~Walther, ``Experimental entanglement of temporal orders,'' {\em
  arXiv:1712.06884}, Dec 2017.

\bibitem{Goswami_2018}
K.~Goswami, C.~Giarmatzi, M.~Kewming, F.~Costa, C.~Branciard, J.~Romero, and
  A.~White, ``Indefinite causal order in a quantum switch,'' {\em Physical
  Review Letters}, vol.~121, Aug 2018.

\bibitem{Guerin_experiment}
K.~Wei, N.~Tischler, S.-R. Zhao, Y.-H. Li, J.~M. Arrazola, Y.~Liu, W.~Zhang,
  H.~Li, L.~You, Z.~Wang, Y.-A. Chen, B.~C. Sanders, Q.~Zhang, G.~J. Pryde,
  F.~Xu, and J.-W. Pan, ``Experimental quantum switching for exponentially
  superior quantum communication complexity,'' {\em Physical Review Letters},
  vol.~122, p.~120504, Mar 2019.

\bibitem{taddei2020experimental}
M.~M. Taddei, J.~Cari\~ne, D.~Mart\'{\i}nez, T.~Garc\'{\i}a, N.~Guerrero, A.~A.
  Abbott, M.~Ara\'ujo, C.~Branciard, E.~S. G\'omez, S.~P. Walborn, L.~Aolita,
  and G.~Lima, ``Computational advantage from the quantum superposition of
  multiple temporal orders of photonic gates,'' {\em PRX Quantum}, vol.~2,
  p.~010320, Feb 2021.

\bibitem{1}
M.~Araújo, F.~Costa, and {\v{C}}.~Brukner, ``Computational advantage from
  quantum-controlled ordering of gates,'' {\em Physical Review Letters},
  vol.~113, Dec 2014.

\bibitem{Colnaghi_2012}
T.~Colnaghi, G.~M. D'Ariano, S.~Facchini, and P.~Perinotti, ``Quantum
  computation with programmable connections between gates,'' {\em Physics
  Letters A}, vol.~376, p.~2940–2943, Oct 2012.

\bibitem{Facchini_2015}
S.~Facchini and S.~Perdrix, ``Quantum circuits for the
  unitary permutation problem,'' {\em Theory and Applications of Models of
  Computation}, p.~324–331, 2015.

\bibitem{Kleitman_1976}
D.~Kleitman and D.~Kwiatkowski, ``A lower bound on the length of a sequence
  containing all permutations as subsequences,'' {\em Journal of Combinatorial
  Theory, Series A}, vol.~21, no.~2, pp.~129 -- 136, 1976.

\bibitem{newey}
M.~Newey, ``Notes on a problem involving permutations as subsequences,'' {\em
  Technical Report 340, Stanford University}, Mar 1973.

\bibitem{zalinescu}
E.~Zalinescu, ``Shorter strings containing all k-element permutations,'' {\em
  Information Processing Letters}, vol.~111, pp.~605--608, Jun 2011.

\bibitem{radomirovic}
S.~Radomirovic, ``A construction of short sequences containing all permutations
  of a set as subsequences,'' {\em The Electronic Journal of Combinatorics},
  vol.~19, Nov 2012.

\bibitem{Ara_jo_2014}
M.~Araújo, A.~Feix, F.~Costa, and {\v{C}}.~Brukner, ``Quantum circuits cannot
  control unknown operations,'' {\em New Journal of Physics}, vol.~16,
  p.~093026, Sep 2014.

\bibitem{cunknown1}
B.~P. Lanyon, M.~Barbieri, M.~P. Almeida, T.~Jennewein, T.~C. Ralph, K.~J.
  Resch, G.~J. Pryde, J.~L. O'Brien, A.~Gilchrist, and A.~G. White,
  ``Simplifying quantum logic using higher-dimensional hilbert spaces,'' {\em
  Nature Physics}, vol.~5, pp.~134--140, Feb 2009.

\bibitem{cunknown2}
X.-Q. Zhou, T.~C. Ralph, P.~Kalasuwan, M.~Zhang, A.~Peruzzo, B.~P. Lanyon, and
  J.~L. O'Brien, ``Adding control to arbitrary unknown quantum operations,''
  {\em Nature Communications}, vol.~2, p.~413, Aug 2011.

\bibitem{cunknown3}
X.-Q. Zhou, P.~Kalasuwan, T.~C. Ralph, and J.~L. O'Brien, ``Calculating unknown
  eigenvalues with a quantum algorithm,'' {\em Nature Photonics}, vol.~7,
  pp.~223--228, Mar 2013.

\bibitem{eastin2004qcircuit}
B.~Eastin and S.~T. Flammia, ``Q-circuit tutorial,'' {\em
  arXiv:quant-ph/0406003}, Aug 2004.

\end{thebibliography}
\bibliographystyle{ieeetr}

\appendix

\begin{widetext}

\section{Representation of $x$ in the algorithm for the FPP with the factoradic labeling} \label{appedixA}
For the algorithm in Section \ref{seclog} to work, it is necessary that every $x\in\{0,1,...,n!-1\}$ can be represented in the basis given in Eq.~\eqref{equlog}. Here, we prove this statement:\\\\
\begin{lemma}

For every $x\in\{0,1,...,n!-1\}$, there exists $(n-1)\cdot \ceil*{\log_2(n)}$ bits $c^x_{k,i}\in\{0,1\}$ such that:
\begin{align}
    x=\sum_{k=1}^{n-1}\sum_{i=1}^{\ceil*{\log_2(n)}}c^x_{k,i}\cdot \ceil*{\frac{k}{2^i}}\cdot k! \ . \label{appalog}
\end{align}

\end{lemma}
\begin{proof}
First, we note that every $x\in\{0,...,n!-1\}$ can be represented in the factorial number system:
\begin{align}
    x=\sum_{k=1}^{n-1}a_k\cdot k! \, ,\label{appafac}
\end{align}
with coefficients $a_k\in \{0,1,...,k\}$. Hence, it is enough to show that for every $1\leq k\leq n-1$, every $a_k\in \{0,...,k\}$ can be written as:
\begin{align}
    a_k=\sum_{i=1}^{\ceil*{\log_2(n)}}c_{k,i}\cdot \ceil*{\frac{k}{2^i}} \, , \label{representak}
\end{align}
where $c_{k,i}\in\{0,1\}$. For a clearer notation, we will drop the index ``$x$'' in $c^x_{k,i}$ from now on. Furthermore, let $1\leq k\leq n-1$ throughout the proof. In a first step, we prove by induction that the bits $c_{k,1}, c_{k,2}, ..., c_{k,\ceil*{\log_2(n)}}$ can represent every number between 0 and $\sum_{i=1}^{\ceil*{\log_2(n)}} \ceil*{\frac{k}{2^i}}$ (according to Eq.~\eqref{representak}). In a second step, we show that $k\leq \sum_{i=1}^{\ceil*{\log_2(n)}} \ceil*{\frac{k}{2^i}}$, from which we conclude that indeed every $a_k\in \{0,...,k\}$ can be written in the above form.\\\\

\textit{Step 1:}\\
For every number $0\leq a_k\leq \sum_{i=1}^{\ceil*{\log_2(n)}} \ceil*{\frac{k}{2^i}}$, there exist bits $c_{k,1}, c_{k,2}, ..., c_{k,\ceil*{\log_2(n)}}\in\{0,1\}$ such that
\begin{align}
    a_k=\sum_{i=1}^{\ceil*{\log_2(n)}}c_{k,i}\cdot \ceil*{\frac{k}{2^i}} \ .\label{representation}
\end{align}
\textbf{Base case:}\\
Every number $0\leq a_k\leq \sum_{i=\ceil*{\log_2(n)}}^{\ceil*{\log_2(n)}} \ceil*{\frac{k}{2^i}}=1$ can be represented with $c_{k,1}=...=c_{k,\ceil*{\log_2(n)}-1}=0$ and $c_{k,\ceil*{\log_2(n)}}=a_k$, since $\ceil*{\frac{k}{2^{\ceil*{\log_2(n)}}}}=1$ whenever $1\leq k\leq n-1$.
\\\\
\textbf{Induction step:}\\
Suppose that $c_{k,1}=...=c_{k,(j-1)}=0$ and the bits $c_{k,j},c_{k,(j+1)}, ..., c_{k,\ceil*{\log_2(n)}}\in\{0,1\}$ are sufficient to represent every number $a_k$ between 0 and $\sum_{i=j}^{\ceil*{\log_2(n)}} \ceil*{\frac{k}{2^i}}$ (according to Eq.~\eqref{representation}). If we now flip the bit $c_{k,(j-1)}$ to one, we can also represent every number $a_k$ between $\ceil*{\frac{k}{2^{j-1}}}$ and $\ceil*{\frac{k}{2^{j-1}}}+\sum_{i=j}^{\ceil*{\log_2(n)}}\ceil*{\frac{k}{2^i}}=\sum_{i=j-1}^{\ceil*{\log_2(n)}}\ceil*{\frac{k}{2^i}}$, although $c_{k,1}=...=c_{k,(j-2)}=0$. To conclude that the bits $c_{k,(j-1)},c_{k,j},c_{k,(j+1)}, ..., c_{k,\ceil*{\log_2(n)}}\in\{0,1\}$ (while all other bits are set to zero) are enough to represent every number between 0 and $\sum_{i=j-1}^{\ceil*{\log_2(n)}}\ceil*{\frac{k}{2^i}}$, we have to show that:
\begin{align}
    \ceil*{\frac{k}{2^{j-1}}}\leq \sum_{i=j}^{\ceil*{\log_2(n)}} \ceil*{\frac{k}{2^i}}+1 \ .
\end{align}
Since $\ceil*{x+y}\leq \ceil*{x}+\ceil*{y}$ for all real numbers $x$ and $y$, this can be inferred from repeating
\begin{align}
    \ceil*{\frac{k}{2^{j-1}}}=\ceil*{\frac{k}{2^{j}}+\frac{k}{2^{j}}}\leq \ceil*{\frac{k}{2^{j}}}+\ceil*{\frac{k}{2^{j}}}
\end{align}
for $j+1$, $j+2$, ..., $\ceil*{\log_2(n)}$:
\begin{align}
    \ceil*{\frac{k}{2^{j-1}}}\leq \ceil*{\frac{k}{2^{j}}}+\ceil*{\frac{k}{2^{j}}}\leq \ceil*{\frac{k}{2^{j}}}+\ceil*{\frac{k}{2^{j+1}}}+\ceil*{\frac{k}{2^{j+1}}}\leq ... \leq \sum_{i=j}^{\ceil*{\log_2(n)}} \ceil*{\frac{k}{2^i}} + \ceil*{\frac{k}{2^{\ceil*{\log_2(n)}}}}=\sum_{i=j}^{\ceil*{\log_2(n)}} \ceil*{\frac{k}{2^i}}+1
\end{align}
and using $\ceil*{\frac{k}{2^{\ceil*{\log_2(n)}}}}=1$ (for $1\leq k\leq n-1$) in the last step.\\

Hence, by induction, we have shown that the bits $c_{k,1}, c_{k,2}, ..., c_{k,\ceil*{\log_2(n)}}$ can represent every number $a_k$ between 0 and $\sum_{i=1}^{\ceil*{\log_2(n)}} \ceil*{\frac{k}{2^i}}$ (according to Eq.~\eqref{representation}).\\

\textit{Step 2:}\\
To show that every $a_k\in\{0,...,k\}$ can be represented with these bits, we have to prove that:
\begin{align}
    \sum_{i=1}^{\ceil*{\log_2(n)}} \ceil*{\frac{k}{2^i}}\geq k  \, .
\end{align}
This is true since:
\begin{align}
\begin{split}
    \sum_{i=1}^{\ceil*{\log_2(n)}} \ceil*{\frac{k}{2^i}} \geq \ceil*{\sum_{i=1}^{\ceil*{\log_2(n)}} \frac{k}{2^i}} = \ceil*{\left(1-\frac{1}{2^{\ceil*{\log_2(n)}}}\right)k}
    \geq \ceil*{k-\frac{k}{n}}\geq k \, ,
\end{split}
\end{align}
where we have used $\sum_i \ceil*{x_i}\geq \ceil*{\sum_i x_i}$ in the first step, the geometric series in the second, $2^{\ceil*{\log_2(n)}}\geq n$ in the third and $\frac{k}{n}<1$ (for $1\leq k\leq n-1$) in the last step. This completes the proof.
\end{proof}

If one is interested in obtaining the bits $c_{k,i}$ for a given $a_k$, one can use recursion: $c_{k,1}=1$ if $a_k \geq \ceil*{\frac{k}{2}}$ and 0 if $a_k < \ceil*{\frac{k}{2}}$. The bit $c_{k,2}$ is 1 if $a_k-c_{k,1}\cdot \ceil*{\frac{k}{2}} \geq \ceil*{\frac{k}{2^2}}$ and 0 if $a_k-c_{k1}\cdot \ceil*{\frac{k}{2}} < \ceil*{\frac{k}{2^2}}$. Following this, the bit $c_{k,j}$ is 1 if $a_k-\sum_{i=1}^{j-1}c_{k,i}\cdot \ceil*{\frac{k}{2^i}}\geq \ceil*{\frac{k}{2^j}}$ and 0 if $a_k-\sum_{i=1}^{j-1}c_{k,i}\cdot \ceil*{\frac{k}{2^i}}< \ceil*{\frac{k}{2^j}}$. Note however that the representation of $x$ (and the corresponding coefficients $a_k$ of $x$) in this basis is not unique and this procedure is not the only one that leads to a representation of $x$ in this basis.

\section{Proof of Lemma~\ref{lemmasqrt}} \label{appedixB}

\setcounter{definition}{1}
\setcounter{lemma}{0}

In this section, we will prove the remaining statement of Section~\ref{secsqrt}. As a help to the reader, we restate the important definition for this lemma.

\begin{definition}
For every permutation \mbox{$\Pi_x=(U_{\sigma_x(n-1)}...U_{\sigma_x((\hat{k}-1)\cdot \hat{n})})...\ 
    (U_{\sigma_x(2\cdot \hat{n}-1)}...U_{\sigma_x(\hat{n})})\ 
    (U_{\sigma_x(\hat{n}-1)}...U_{\sigma_x(0)})$} of the $n$ unitaries, let
\begin{align}
\Pi_{xk}&:=
\begin{cases}
    [U_{\sigma_x(n-1)}...U_{\sigma_x(\hat{n})}]\ 
    (U_{\sigma_x(\hat{n}-1)}...U_{\sigma_x(0)})\\
    [U_{\sigma_x(n-1)}...U_{\sigma_x((k+1)\cdot \hat{n})}]\ 
    (U_{\sigma_x((k+1)\cdot \hat{n}-1)}...U_{\sigma_x(k\cdot \hat{n})})\ 
    [U_{\sigma_x(k\cdot \hat{n}-1)}...U_{\sigma_x(0)}]\\
    (U_{\sigma_x(n-1)}...U_{\sigma_x((\hat{k}-1)\cdot \hat{n})})\ 
    [U_{\sigma_x((\hat{k}-1)\cdot \hat{n}-1)}...U_{\sigma_x(0)}]
\end{cases}
&
\begin{array}{l}
    k=0\\
    k=1,...,\hat{k}-2\\
    k=\hat{k}-1
\end{array} \label{pixkapp}
\\
\tilde{\Pi}^r_{xk}&:=
\begin{cases}
    \{U_{\sigma_x(k\cdot \hat{n}-1)}...U_{\sigma_x(0)}\}\ 
    \{U_{\sigma_x(n-1)}...U_{\sigma_x(k\cdot \hat{n})}\}\\
    \{U_{\sigma_x((\hat{k}-1)\cdot \hat{n}-1)}...U_{\sigma_x(0)}\}\ 
    \{U_{\sigma_x(n-1)}...U_{\sigma_x((\hat{k}-1)\cdot \hat{n})}\}
\end{cases}
&
\begin{array}{l}
    k=1,...,\hat{k}-2\\
    k=\hat{k}-1
\end{array} \label{pixktildeapp}
\end{align}
where $[U_{i_1}U_{i_2}...U_{i_j}]$ is defined to be the descending ordering of the unitaries $U_{i_1}U_{i_2}...U_{i_j}$, while $\{U_{i_1}U_{i_2}...U_{i_j}\}$ is the ascending ordering of them and $(U_{i_1}U_{i_2}...U_{i_j})=U_{i_1}U_{i_2}...U_{i_j}$ leaves the string invariant. (\mbox{$\hat{n}:=\ceil*{\sqrt{n}}$} and \mbox{$\hat{k}:=\ceil*{\frac{n}{\hat{n}}}$})

\end{definition}

\begin{lemma}
For every set of ($d$-dimensional) unitaries $\{U_i\}_0^{n-1}$ that satisfy pairwise commutation relations, the following relation holds:
\begin{align}
\begin{split}
    &\bigotimes^{\hat{k}-1}_{k=0}\Pi_{xk}\ket{\Psi_k}\otimes \bigotimes^{\hat{k}-1}_{k=1}\tilde{\Pi}^r_{xk}\ket{\Phi_k}=\Pi_x\ket{\Psi_0}\otimes \bigotimes^{\hat{k}-1}_{k=1}\Pi\ket{\Psi_k}\otimes \Pi^r\ket{\Phi_k} \, .
    \label{toproofsqrt}
\end{split}
\end{align}
Here, $\Pi:=U_{n-1}...U_1U_0$ denotes the descending order of all unitaries, $\Pi^r:=U_{0}U_1...U_{n-1}$ denotes the ascending order of all unitaries and $\ket{\Psi_k},\ket{\Phi_k}\in \mathcal{H}^d$ are arbitrary $d$-dimensional states.
\end{lemma}

\begin{proof}
We will prove this by expressing both sides in terms of pairwise phases and comparing them at the end.
Recall that the pairwise phase $\alpha_{jk}$ is defined via $U_jU_k=\alpha_{jk}\ U_kU_j$ and from comparing this with $U_kU_j=\alpha_{kj}\ U_jU_k$, we obtain $\alpha_{jk}=(\alpha_{kj})^{-1}$.\\

The total phase of the permutation $\Pi_x$ is the product of the following pairwise phases $\alpha_{ij}$:
\begin{align}
    \Pi_x=\prod_{\substack{n-1\geq i>j\geq 0\\\text{with } \sigma_x(i)<\sigma_x(j) }}\alpha_{\sigma_x(i)\sigma_x(j)}\cdot \Pi \, .
    \label{prod}
\end{align}
To obtain this expression, we can compare every pair of unitaries and check whether they have the same order as in $\Pi$. If not, the corresponding phase appears in the product. To be more precise, we can start with $j=0$ (with the matrix $U_{\sigma(0)}$) and check for every $i>j=0$ if $\sigma_x(i)<\sigma_x(j=0)$. If so, then the order of $U_{\sigma_x(0)}$ and $U_{\sigma_x(i)}$ is reversed between $\Pi_x$ and $\Pi$ and the corresponding phase appears as a factor in the total phase of $\Pi_x$ with respect to $\Pi$. By repeating this procedure for every $n-1\geq j\geq 0$, we obtain the above expression. For example, for $n=4$ and the permutation $\Pi_x=U_1U_2U_0U_3$ we obtain:
\begin{align}
    U_1U_2U_0U_3=\alpha_{03}\cdot\alpha_{23}\cdot\alpha_{13}\cdot\alpha_{12}\cdot U_3U_2U_1U_0 \, .
\end{align}

For the left hand side of the statement, we can compute for every $\Pi_{xk}$ the relative phase to $\Pi$ and for every $\tilde{\Pi}^r_{xk}$ the relative phase to $\Pi^r$ in terms of the pairwise phases $\alpha_{ij}$. One could calculate these phases in a direct way. Here we will follow this approach but with some shortcuts; in a first step, we rewrite for every $k\in\{0,1,.., \hat{k}-2\}$ the first two blocks of $\Pi_{xk}$, namely $[U_{\sigma_x(n-1)}...U_{\sigma_x((k+1)\cdot \hat{n})}]\ (U_{\sigma_x((k+1)\cdot \hat{n}-1)}...U_{\sigma_x(k\cdot \hat{n})})$, into:
\begin{align}
    [U_{\sigma_x(n-1)}...U_{\sigma_x((k+1)\cdot \hat{n})}]\ 
    (U_{\sigma_x((k+1)\cdot \hat{n}-1)}...U_{\sigma_x(k\cdot \hat{n})})=\prod_{\substack{(k+1)\cdot \hat{n}>j\geq k\cdot \hat{n}\\ \text{and } n-1\geq i>j,\\\text{with } \sigma_x(i)<\sigma_x(j) }}\alpha_{\sigma_x(i)\sigma_x(j)}\cdot [U_{\sigma_x(n-1)}...U_{\sigma_x(k\cdot \hat{n})}]\, .\\
\end{align}
To see that this is true, remember that we obtain the phase as a product of the pairwise phases by comparing each pair of positions $n-1\geq i>j\geq k\cdot \hat{n}$. But since the left block $[U_{\sigma_x(n-1)}...U_{\sigma_x((k+1)\cdot \hat{n})}]$ is already ordered, we do not have to consider the cases for which $n-1\geq i>j\geq (k+1)\cdot \hat{n}$. Similarly for $k=\hat{k}-1$, we rewrite:
\begin{align}
    (U_{\sigma_x(n-1)}...U_{\sigma_x((\hat{k}-1)\cdot \hat{n})})=
    \prod_{\substack{n-1\geq j\geq (\hat{k}-1)\cdot \hat{n}\\ \text{and } n-1\geq i>j\\\text{with } \sigma_x(i)<\sigma_x(j) }}\alpha_{\sigma_x(i)\sigma_x(j)}\cdot [U_{\sigma_x(n-1)}...U_{\sigma_x((\hat{k}-1)\cdot \hat{n})}] \, .
\end{align}

To summarize, we can rewrite the permutations $\Pi_{xk}$ into:
\begin{align}
    \Pi_{xk}=\alpha_{xk}\cdot \tilde{\Pi}_{xk} \, ,
    \label{B1}
\end{align}
where:
\begin{align*}
\tilde{\Pi}_{xk}=
\begin{cases}
    [U_{\sigma_x(n-1)}...U_{\sigma_x(0)}]&k=0\\
    [U_{\sigma_x(n-1)}...U_{\sigma_x(k\cdot \hat{n})}]\ 
    [U_{\sigma_x(k\cdot \hat{n}-1)}...U_{\sigma_x(0)}]&k=1,...,\hat{k}-2\\
    [U_{\sigma_x(n-1)}...U_{\sigma_x((\hat{k}-1)\cdot \hat{n})}]\ 
    [U_{\sigma_x((\hat{k}-1)\cdot \hat{n}-1)}...U_{\sigma_x(0)}]&k=\hat{k}
    -1
\end{cases}
\end{align*}
and:
\begin{align*}
\alpha_{xk}=\prod_{\substack{(k+1)\cdot \hat{n}>j\geq k\cdot \hat{n}\\ \text{and } n-1\geq i>j,\\\text{with } \sigma_x(i)<\sigma_x(j) }}\alpha_{\sigma_x(i)\sigma_x(j)}\text{ for every }k\in\{0,1,...,\hat{k}-2\}&&\text{ and }&&
    \alpha_{x(\hat{k}-1)}=\prod_{\substack{n-1\geq j\geq (\hat{k}-1)\cdot \hat{n}\\ \text{and } n-1\geq i>j\\\text{with } \sigma_x(i)<\sigma_x(j) }}\alpha_{\sigma_x(i)\sigma_x(j)} \, .
\end{align*}
The product of all the $\alpha_{xk}$ is exactly the same expression as in Eq.~\eqref{prod}:
\begin{align}
    \prod^{\hat{k}-1}_{k= 0}\alpha_{xk}
    =\prod_{\substack{n-1\geq j\geq 0\\ \text{and } n-1\geq i>j\\\text{with } \sigma_x(i)<\sigma_x(j) }}\alpha_{\sigma_x(i)\sigma_x(j)} \, .
    \label{B3}
\end{align}
In this way, we already obtain all required phases.
For the above example of $\Pi_x=U_1U_2\ U_0U_3$, these expressions read:
\begin{align}
    \Pi_{x0}&=U_2U_1\ U_0U_3=\alpha_{03}\cdot\alpha_{13}\cdot\alpha_{23}\cdot U_3U_2\ U_1U_0&&(\tilde{\Pi}_{x0}=U_3U_2\ U_1U_0)\\
    \Pi_{x1}&=U_1U_2\ U_3U_0=\alpha_{12}\cdot U_2U_1\ U_3U_0&&(\tilde{\Pi}_{x1}=U_2U_1\ U_3U_0)\\
    &&&(\tilde{\Pi}^r_{x1}=U_0U_3\ U_1U_2) \, .
\end{align}

Note that for $k=0$, the permutation $\tilde{\Pi}_{x0}=[U_{\sigma_x(n-1)}...U_{\sigma_x(0)}]$ always equals $\Pi$. On the other hand, the permutations $\tilde{\Pi}_{xk}$ (for $k\geq 1$) would lead in general to additional (unnecessary) phases relative to $\Pi$ but they are exactly compensated by the phases of $\tilde{\Pi}_{xk}^r$ relative to $\Pi^r$. To see this, note that $\tilde{\Pi}_{xk}$ and $\tilde{\Pi}_{xk}^r$ are (by construction) the same permutations but in reversed order. This has the property that the relative phase between $\tilde{\Pi}_{xk}$ and $\Pi$ is exactly the inverse of the relative phase between $\tilde{\Pi}_{xk}^r$ and $\Pi^r$. This is true since, whenever two unitaries $U_i$ and $U_j$ are commuted in $\tilde{\Pi}_{xk}$ relative to $\Pi$ (and we obtain $\alpha_{ij}$ as a factor in the relative phase between $\tilde{\Pi}_{xk}$ and $\Pi$), then the two unitaries are also commuted in $\tilde{\Pi}_{xk}^r$ relative to $\Pi^r$ (and we obtain the phase $\alpha_{ji}=(\alpha_{ij})^{-1}$ as a factor in the relative phase between $\tilde{\Pi}_{xk}^r$ and $\Pi^r$). For the example of $n=4$ and $\Pi_x=U_1U_2U_0U_3$ this becomes:
\begin{align}
\begin{split}
    \tilde{\Pi}_{x1}=U_2U_1\ U_3U_0=&\alpha_{13}\cdot\alpha_{23}\cdot U_3U_2\ U_1U_0\\
    \tilde{\Pi}^r_{x1}=U_0U_3\ U_1U_2=&\alpha_{31}\cdot\alpha_{32}\cdot U_0U_1\ U_2U_3=(\alpha_{13})^{-1}\cdot(\alpha_{23})^{-1}\cdot U_0U_1\ U_2U_3\\\\
    \implies \tilde{\Pi}_{x1}\ket{\Psi_1}\otimes \tilde{\Pi}^r_{x1}\ket{\Phi_1}=&U_2U_1\ U_3U_0\ket{\Psi_1}\otimes U_0U_3\ U_1U_2\ket{\Phi_1}\\
    =&U_3U_2\ U_1U_0\ket{\Psi_1}\otimes U_0U_1\ U_2U_3\ket{\Phi_1}=\Pi\ket{\Psi_1}\otimes \Pi^r\ket{\Phi_1}     \, .
\end{split}
\end{align}
In this way, we obtain
\begin{align}
    \text{for }k=0:&&\tilde{\Pi}_{x0}=\Pi \, ,&&&&\text{ and for every }k\in\{1,2,...,\hat{k}-1\}:&&\tilde{\Pi}_{xk}\ket{\Psi_k}\otimes \tilde{\Pi}^r_{xk}\ket{\Phi_k}=\Pi\ket{\Psi_k}\otimes \Pi^r\ket{\Phi_k} \, . \label{B2}
\end{align}
Putting everything together, we obtain the desired equation:
\begin{align}
    \bigotimes^{\hat{k}-1}_{k= 0}\Pi_{xk}\ket{\Psi_k}\otimes \bigotimes^{\hat{k}-1}_{k= 1}\tilde{\Pi}^r_{xk}\ket{\Phi_k}&=\bigotimes^{\hat{k}-1}_{k= 0}\alpha_{xk}\cdot \tilde{\Pi}_{xk}\ket{\Psi_k}\otimes\bigotimes^{\hat{k}-1}_{k= 1} \tilde{\Pi}^r_{xk}\ket{\Phi_k}\\
    &=\left(\prod^{\hat{k}-1}_{k= 0}\alpha_{xk}\right)\cdot \tilde{\Pi}_{x0}\ket{\Psi_0}\otimes\bigotimes^{\hat{k}-1}_{k= 1} \tilde{\Pi}_{xk}\ket{\Psi_k}\otimes\tilde{\Pi}^r_{xk}\ket{\Phi_k}\\
    &=\left(\prod^{\hat{k}-1}_{k= 0}\alpha_{xk}\right)\cdot \Pi\ket{\Psi_0}\otimes\bigotimes^{\hat{k}-1}_{k= 1} \Pi\ket{\Psi_k}\otimes\Pi^r\ket{\Phi_k}\\
    &=\left(\prod_{\substack{n-1\geq j\geq 0\\ \text{and } n-1\geq i>j\\\text{with } \sigma_x(i)<\sigma_x(j) }}\alpha_{\sigma_x(i)\sigma_x(j)}\right)\cdot \Pi\ket{\Psi_0}\otimes\bigotimes^{\hat{k}-1}_{k= 1} \Pi\ket{\Psi_k}\otimes\Pi^r\ket{\Phi_k}\\
    &=\Pi_x\ket{\Psi_0}\otimes\bigotimes^{\hat{k}-1}_{k= 1} \Pi\ket{\Psi_k}\otimes\Pi^r\ket{\Phi_k}
\end{align}
where we used Eq.~\eqref{B1} in the first step, Eq.~\eqref{B2} in the third step, Eq.~\eqref{B3} in the fourth step and Eq.~\eqref{prod} in the last one. This completes the proof.
\end{proof}

\newpage
\section{$O(n\log{n})$-algorithm for $n=8$} \label{appedixC}
Here, we present the algorithm that solves the FPP with the factoradic labeling of the permutations for $n=8$. It consists of 14 $d$-dimensional target systems and 21 control qubits. The gate $U_0$ appears once on every target system and the gates $U_1, ..., U_7$ appear each six times. Hence, in total 56 black-box unitaries are used. Note that the shortest known string containing all 8! permutations of the eight unitaries $U_0, ..., U_7$ has length 51 \cite{zalinescu}. Nevertheless, if we omit the target systems $\ket{\Psi_{8,1}}$, $\ket{\Psi_{8,2}}$, $\ket{\Psi_{8,3}}$ and $\ket{\Psi_{8,8}}$ from our circuit together with the involved black-box unitaries and the corresponding control qubits $\ket{c_{1,3}}$, $\ket{c_{2,3}}$ and $\ket{c_{3,3}}$ (since they are unnecessary to represent every number $x\in\{0,...,8!-1\}$), we save 10 queries and solve the problem by calling only 46 black-box unitaries.

\newpage
\smaller[2]
\rotatebox{90}{
$$\Qcircuit @C=0.28em @R=0.6em {
\lstick{\ket{c_{7,1}}}    & \qw & \ctrl{21}  & \qw        & \qw        & \qw        & \qw        & \qw        & \qw        & \qw        & \qw        & \qw        & \qw        & \qw        & \qw        & \qw        & \qw        & \qw        & \qw        & \qw        & \qw        & \qw        & \qw        & \qw          & \qw        & \qw        & \qw        & \qw        & \qw        & \qw        & \qw        & \qw        & \qw        & \qw        & \qw        & \qw        & \qw        & \qw        & \qw        & \qw        & \qw        & \qw        & \qw        & \qw        & \ctrlo{21} & \qw & \rstick{(\omega^{4\cdot 7!\cdot y})} \\
\lstick{\ket{c_{7,2}}}    & \qw & \qw        & \ctrl{24}  & \qw        & \qw        & \qw        & \qw        & \qw        & \qw        & \qw        & \qw        & \qw        & \qw        & \qw        & \qw        & \qw        & \qw        & \qw        & \qw        & \qw        & \qw        & \qw        & \qw          & \qw        & \qw        & \qw        & \qw        & \qw        & \qw        & \qw        & \qw        & \qw        & \qw        & \qw        & \qw        & \qw        & \qw        & \qw        & \qw        & \qw        & \qw        & \qw        & \ctrlo{24} & \qw        & \qw & \rstick{(\omega^{2\cdot 7!\cdot y})} \\
\lstick{\ket{c_{7,3}}}    & \qw & \qw        & \qw        & \ctrl{31}  & \qw        & \qw        & \qw        & \qw        & \qw        & \qw        & \qw        & \qw        & \qw        & \qw        & \qw        & \qw        & \qw        & \qw        & \qw        & \qw        & \qw        & \qw        & \qw          & \qw        & \qw        & \qw        & \qw        & \qw        & \qw        & \qw        & \qw        & \qw        & \qw        & \qw        & \qw        & \qw        & \qw        & \qw        & \qw        & \qw        & \qw        & \ctrlo{31} & \qw        & \qw        & \qw & \rstick{(\omega^{1\cdot 7!\cdot y})} \\
\lstick{\ket{c_{6,1}}}    & \qw & \qw        & \qw        & \qw        & \ctrl{19}  & \qw        & \qw        & \qw        & \qw        & \qw        & \qw        & \qw        & \qw        & \qw        & \qw        & \qw        & \qw        & \qw        & \qw        & \qw        & \qw        & \qw        & \qw          & \qw        & \qw        & \qw        & \qw        & \qw        & \qw        & \qw        & \qw        & \qw        & \qw        & \qw        & \qw        & \qw        & \qw        & \qw        & \qw        & \qw        & \ctrlo{19} & \qw        & \qw        & \qw        & \qw & \rstick{(\omega^{3\cdot 6!\cdot y})} \\
\lstick{\ket{c_{6,2}}}    & \qw & \qw        & \qw        & \qw        & \qw        & \ctrl{20}  & \qw        & \qw        & \qw        & \qw        & \qw        & \qw        & \qw        & \qw        & \qw        & \qw        & \qw        & \qw        & \qw        & \qw        & \qw        & \qw        & \qw          & \qw        & \qw        & \qw        & \qw        & \qw        & \qw        & \qw        & \qw        & \qw        & \qw        & \qw        & \qw        & \qw        & \qw        & \qw        & \qw        & \ctrlo{20} & \qw        & \qw        & \qw        & \qw        & \qw & \rstick{(\omega^{2\cdot 6!\cdot y})} \\
\lstick{\ket{c_{6,3}}}    & \qw & \qw        & \qw        & \qw        & \qw        & \qw        & \ctrl{27}  & \qw        & \qw        & \qw        & \qw        & \qw        & \qw        & \qw        & \qw        & \qw        & \qw        & \qw        & \qw        & \qw        & \qw        & \qw        & \qw          & \qw        & \qw        & \qw        & \qw        & \qw        & \qw        & \qw        & \qw        & \qw        & \qw        & \qw        & \qw        & \qw        & \qw        & \qw        & \ctrlo{27} & \qw        & \qw        & \qw        & \qw        & \qw        & \qw & \rstick{(\omega^{1\cdot 6!\cdot y})} \\
\lstick{\ket{c_{5,1}}}    & \qw & \qw        & \qw        & \qw        & \qw        & \qw        & \qw        & \ctrl{15}  & \qw        & \qw        & \qw        & \qw        & \qw        & \qw        & \qw        & \qw        & \qw        & \qw        & \qw        & \qw        & \qw        & \qw        & \qw          & \qw        & \qw        & \qw        & \qw        & \qw        & \qw        & \qw        & \qw        & \qw        & \qw        & \qw        & \qw        & \qw        & \qw        & \ctrlo{15} & \qw        & \qw        & \qw        & \qw        & \qw        & \qw        & \qw & \rstick{(\omega^{3\cdot 5!\cdot y})} \\
\lstick{\ket{c_{5,2}}}    & \qw & \qw        & \qw        & \qw        & \qw        & \qw        & \qw        & \qw        & \ctrl{16}  & \qw        & \qw        & \qw        & \qw        & \qw        & \qw        & \qw        & \qw        & \qw        & \qw        & \qw        & \qw        & \qw        & \qw          & \qw        & \qw        & \qw        & \qw        & \qw        & \qw        & \qw        & \qw        & \qw        & \qw        & \qw        & \qw        & \qw        & \ctrlo{16} & \qw        & \qw        & \qw        & \qw        & \qw        & \qw        & \qw        & \qw & \rstick{(\omega^{2\cdot 5!\cdot y})} \\
\lstick{\ket{c_{5,3}}}    & \qw & \qw        & \qw        & \qw        & \qw        & \qw        & \qw        & \qw        & \qw        & \ctrl{23}  & \qw        & \qw        & \qw        & \qw        & \qw        & \qw        & \qw        & \qw        & \qw        & \qw        & \qw        & \qw        & \qw          & \qw        & \qw        & \qw        & \qw        & \qw        & \qw        & \qw        & \qw        & \qw        & \qw        & \qw        & \qw        & \ctrlo{23} & \qw        & \qw        & \qw        & \qw        & \qw        & \qw        & \qw        & \qw        & \qw & \rstick{(\omega^{1\cdot 5!\cdot y})} \\
\lstick{\ket{c_{4,1}}}    & \qw & \qw        & \qw        & \qw        & \qw        & \qw        & \qw        & \qw        & \qw        & \qw        & \ctrl{13}  & \qw        & \qw        & \qw        & \qw        & \qw        & \qw        & \qw        & \qw        & \qw        & \qw        & \qw        & \qw          & \qw        & \qw        & \qw        & \qw        & \qw        & \qw        & \qw        & \qw        & \qw        & \qw        & \qw        & \ctrlo{13} & \qw        & \qw        & \qw        & \qw        & \qw        & \qw        & \qw        & \qw        & \qw        & \qw & \rstick{(\omega^{2\cdot 4!\cdot y})} \\
\lstick{\ket{c_{4,2}}}    & \qw & \qw        & \qw        & \qw        & \qw        & \qw        & \qw        & \qw        & \qw        & \qw        & \qw        & \ctrl{16}  & \qw        & \qw        & \qw        & \qw        & \qw        & \qw        & \qw        & \qw        & \qw        & \qw        & \qw          & \qw        & \qw        & \qw        & \qw        & \qw        & \qw        & \qw        & \qw        & \qw        & \qw        & \ctrlo{16} & \qw        & \qw        & \qw        & \qw        & \qw        & \qw        & \qw        & \qw        & \qw        & \qw        & \qw & \rstick{(\omega^{1\cdot 4!\cdot y})} \\
\lstick{\ket{c_{4,3}}}    & \qw & \qw        & \qw        & \qw        & \qw        & \qw        & \qw        & \qw        & \qw        & \qw        & \qw        & \qw        & \ctrl{19}  & \qw        & \qw        & \qw        & \qw        & \qw        & \qw        & \qw        & \qw        & \qw        & \qw          & \qw        & \qw        & \qw        & \qw        & \qw        & \qw        & \qw        & \qw        & \qw        & \ctrlo{19} & \qw        & \qw        & \qw        & \qw        & \qw        & \qw        & \qw        & \qw        & \qw        & \qw        & \qw        & \qw & \rstick{(\omega^{1\cdot 4!\cdot y})} \\
\lstick{\ket{c_{3,1}}}    & \qw & \qw        & \qw        & \qw        & \qw        & \qw        & \qw        & \qw        & \qw        & \qw        & \qw        & \qw        & \qw        & \ctrl{9}   & \qw        & \qw        & \qw        & \qw        & \qw        & \qw        & \qw        & \qw        & \qw          & \qw        & \qw        & \qw        & \qw        & \qw        & \qw        & \qw        & \qw        & \ctrlo{9}  & \qw        & \qw        & \qw        & \qw        & \qw        & \qw        & \qw        & \qw        & \qw        & \qw        & \qw        & \qw        & \qw & \rstick{(\omega^{2\cdot 3!\cdot y})} \\
\lstick{\ket{c_{3,2}}}    & \qw & \qw        & \qw        & \qw        & \qw        & \qw        & \qw        & \qw        & \qw        & \qw        & \qw        & \qw        & \qw        & \qw        & \ctrl{12}  & \qw        & \qw        & \qw        & \qw        & \qw        & \qw        & \qw        & \qw          & \qw        & \qw        & \qw        & \qw        & \qw        & \qw        & \qw        & \ctrlo{12} & \qw        & \qw        & \qw        & \qw        & \qw        & \qw        & \qw        & \qw        & \qw        & \qw        & \qw        & \qw        & \qw        & \qw & \rstick{(\omega^{1\cdot 3!\cdot y})} \\
\lstick{\ket{c_{3,3}}}    & \qw & \qw        & \qw        & \qw        & \qw        & \qw        & \qw        & \qw        & \qw        & \qw        & \qw        & \qw        & \qw        & \qw        & \qw        & \ctrl{15}  & \qw        & \qw        & \qw        & \qw        & \qw        & \qw        & \qw          & \qw        & \qw        & \qw        & \qw        & \qw        & \qw        & \ctrlo{15} & \qw        & \qw        & \qw        & \qw        & \qw        & \qw        & \qw        & \qw        & \qw        & \qw        & \qw        & \qw        & \qw        & \qw        & \qw & \rstick{(\omega^{1\cdot 3!\cdot y})} \\
\lstick{\ket{c_{2,1}}}    & \qw & \qw        & \qw        & \qw        & \qw        & \qw        & \qw        & \qw        & \qw        & \qw        & \qw        & \qw        & \qw        & \qw        & \qw        & \qw        & \ctrl{7}   & \qw        & \qw        & \qw        & \qw        & \qw        & \qw          & \qw        & \qw        & \qw        & \qw        & \qw        & \ctrlo{7}  & \qw        & \qw        & \qw        & \qw        & \qw        & \qw        & \qw        & \qw        & \qw        & \qw        & \qw        & \qw        & \qw        & \qw        & \qw        & \qw & \rstick{(\omega^{1\cdot 2!\cdot y})} \\
\lstick{\ket{c_{2,2}}}    & \qw & \qw        & \qw        & \qw        & \qw        & \qw        & \qw        & \qw        & \qw        & \qw        & \qw        & \qw        & \qw        & \qw        & \qw        & \qw        & \qw        & \ctrl{8}   & \qw        & \qw        & \qw        & \qw        & \qw          & \qw        & \qw        & \qw        & \qw        & \ctrlo{8}  & \qw        & \qw        & \qw        & \qw        & \qw        & \qw        & \qw        & \qw        & \qw        & \qw        & \qw        & \qw        & \qw        & \qw        & \qw        & \qw        & \qw & \rstick{(\omega^{1\cdot 2!\cdot y})} \\
\lstick{\ket{c_{2,3}}}    & \qw & \qw        & \qw        & \qw        & \qw        & \qw        & \qw        & \qw        & \qw        & \qw        & \qw        & \qw        & \qw        & \qw        & \qw        & \qw        & \qw        & \qw        & \ctrl{11}  & \qw        & \qw        & \qw        & \qw          & \qw        & \qw        & \qw        & \ctrlo{11} & \qw        & \qw        & \qw        & \qw        & \qw        & \qw        & \qw        & \qw        & \qw        & \qw        & \qw        & \qw        & \qw        & \qw        & \qw        & \qw        & \qw        & \qw & \rstick{(\omega^{1\cdot 2!\cdot y})} \\
\lstick{\ket{c_{1,1}}}    & \qw & \qw        & \qw        & \qw        & \qw        & \qw        & \qw        & \qw        & \qw        & \qw        & \qw        & \qw        & \qw        & \qw        & \qw        & \qw        & \qw        & \qw        & \qw        & \ctrl{3}   & \qw        & \qw        & \qw          & \qw        & \qw        & \ctrlo{3}  & \qw        & \qw        & \qw        & \qw        & \qw        & \qw        & \qw        & \qw        & \qw        & \qw        & \qw        & \qw        & \qw        & \qw        & \qw        & \qw        & \qw        & \qw        & \qw & \rstick{(\omega^{1\cdot 1!\cdot y})} \\
\lstick{\ket{c_{1,2}}}    & \qw & \qw        & \qw        & \qw        & \qw        & \qw        & \qw        & \qw        & \qw        & \qw        & \qw        & \qw        & \qw        & \qw        & \qw        & \qw        & \qw        & \qw        & \qw        & \qw        & \ctrl{4}   & \qw        & \qw          & \qw        & \ctrlo{4}  & \qw        & \qw        & \qw        & \qw        & \qw        & \qw        & \qw        & \qw        & \qw        & \qw        & \qw        & \qw        & \qw        & \qw        & \qw        & \qw        & \qw        & \qw        & \qw        & \qw & \rstick{(\omega^{1\cdot 1!\cdot y})} \\
\lstick{\ket{c_{1,3}}}    & \qw & \qw        & \qw        & \qw        & \qw        & \qw        & \qw        & \qw        & \qw        & \qw        & \qw        & \qw        & \qw        & \qw        & \qw        & \qw        & \qw        & \qw        & \qw        & \qw        & \qw        & \ctrl{7}   & \qw          & \ctrlo{7}  & \qw        & \qw        & \qw        & \qw        & \qw        & \qw        & \qw        & \qw        & \qw        & \qw        & \qw        & \qw        & \qw        & \qw        & \qw        & \qw        & \qw        & \qw        & \qw        & \qw        & \qw & \rstick{(\omega^{1\cdot 1!\cdot y})} \\
\lstick{\ket{\Psi_{2,1}}} & \qw & \gate{U_7} & \qw        & \qw        & \qw        & \qw        & \qw        & \gate{U_5} & \qw        & \qw        & \qw        & \qw        & \qw        & \gate{U_3} & \qw        & \qw        & \qw        & \qw        & \qw        & \gate{U_1} & \qw        & \qw        & \gate{U_{0}} & \qw        & \qw        & \gate{U_1} & \qw        & \qw        & \qw        & \qw        & \qw        & \gate{U_3} & \qw        & \qw        & \qw        & \qw        & \qw        & \gate{U_5} & \qw        & \qw        & \qw        & \qw        & \qw        & \gate{U_7} & \qw &                                      \\
\lstick{\ket{\Psi_{2,2}}} & \qw & \qw        & \qw        & \qw        & \gate{U_6} & \qw        & \qw        & \qw        & \qw        & \qw        & \gate{U_4} & \qw        & \qw        & \qw        & \qw        & \qw        & \gate{U_2} & \qw        & \qw        & \qw        & \qw        & \qw        & \gate{U_{0}} & \qw        & \qw        & \qw        & \qw        & \qw        & \gate{U_2} & \qw        & \qw        & \qw        & \qw        & \qw        & \gate{U_4} & \qw        & \qw        & \qw        & \qw        & \qw        & \gate{U_6} & \qw        & \qw        & \qw        & \qw &                                      \\
\lstick{\ket{\Psi_{4,1}}} & \qw & \qw        & \qw        & \qw        & \qw        & \qw        & \qw        & \qw        & \gate{U_5} & \qw        & \qw        & \qw        & \qw        & \qw        & \qw        & \qw        & \qw        & \qw        & \qw        & \qw        & \gate{U_1} & \qw        & \gate{U_{0}} & \qw        & \gate{U_1} & \qw        & \qw        & \qw        & \qw        & \qw        & \qw        & \qw        & \qw        & \qw        & \qw        & \qw        & \gate{U_5} & \qw        & \qw        & \qw        & \qw        & \qw        & \qw        & \qw        & \qw &                                      \\
\lstick{\ket{\Psi_{4,2}}} & \qw & \qw        & \qw        & \qw        & \qw        & \gate{U_6} & \qw        & \qw        & \qw        & \qw        & \qw        & \qw        & \qw        & \qw        & \qw        & \qw        & \qw        & \gate{U_2} & \qw        & \qw        & \qw        & \qw        & \gate{U_{0}} & \qw        & \qw        & \qw        & \qw        & \gate{U_2} & \qw        & \qw        & \qw        & \qw        & \qw        & \qw        & \qw        & \qw        & \qw        & \qw        & \qw        & \gate{U_6} & \qw        & \qw        & \qw        & \qw        & \qw &                                      \\
\lstick{\ket{\Psi_{4,3}}} & \qw & \qw        & \gate{U_7} & \qw        & \qw        & \qw        & \qw        & \qw        & \qw        & \qw        & \qw        & \qw        & \qw        & \qw        & \gate{U_3} & \qw        & \qw        & \qw        & \qw        & \qw        & \qw        & \qw        & \gate{U_{0}} & \qw        & \qw        & \qw        & \qw        & \qw        & \qw        & \qw        & \gate{U_3} & \qw        & \qw        & \qw        & \qw        & \qw        & \qw        & \qw        & \qw        & \qw        & \qw        & \qw        & \gate{U_7} & \qw        & \qw &                                      \\
\lstick{\ket{\Psi_{4,4}}} & \qw & \qw        & \qw        & \qw        & \qw        & \qw        & \qw        & \qw        & \qw        & \qw        & \qw        & \gate{U_4} & \qw        & \qw        & \qw        & \qw        & \qw        & \qw        & \qw        & \qw        & \qw        & \qw        & \gate{U_{0}} & \qw        & \qw        & \qw        & \qw        & \qw        & \qw        & \qw        & \qw        & \qw        & \qw        & \gate{U_4} & \qw        & \qw        & \qw        & \qw        & \qw        & \qw        & \qw        & \qw        & \qw        & \qw        & \qw &                                      \\
\lstick{\ket{\Psi_{8,1}}} & \qw & \qw        & \qw        & \qw        & \qw        & \qw        & \qw        & \qw        & \qw        & \qw        & \qw        & \qw        & \qw        & \qw        & \qw        & \qw        & \qw        & \qw        & \qw        & \qw        & \qw        & \gate{U_1} & \gate{U_{0}} & \gate{U_1} & \qw        & \qw        & \qw        & \qw        & \qw        & \qw        & \qw        & \qw        & \qw        & \qw        & \qw        & \qw        & \qw        & \qw        & \qw        & \qw        & \qw        & \qw        & \qw        & \qw        & \qw &                                      \\
\lstick{\ket{\Psi_{8,2}}} & \qw & \qw        & \qw        & \qw        & \qw        & \qw        & \qw        & \qw        & \qw        & \qw        & \qw        & \qw        & \qw        & \qw        & \qw        & \qw        & \qw        & \qw        & \gate{U_2} & \qw        & \qw        & \qw        & \gate{U_{0}} & \qw        & \qw        & \qw        & \gate{U_2} & \qw        & \qw        & \qw        & \qw        & \qw        & \qw        & \qw        & \qw        & \qw        & \qw        & \qw        & \qw        & \qw        & \qw        & \qw        & \qw        & \qw        & \qw &                                      \\
\lstick{\ket{\Psi_{8,3}}} & \qw & \qw        & \qw        & \qw        & \qw        & \qw        & \qw        & \qw        & \qw        & \qw        & \qw        & \qw        & \qw        & \qw        & \qw        & \gate{U_3} & \qw        & \qw        & \qw        & \qw        & \qw        & \qw        & \gate{U_{0}} & \qw        & \qw        & \qw        & \qw        & \qw        & \qw        & \gate{U_3} & \qw        & \qw        & \qw        & \qw        & \qw        & \qw        & \qw        & \qw        & \qw        & \qw        & \qw        & \qw        & \qw        & \qw        & \qw &                                      \\
\lstick{\ket{\Psi_{8,4}}} & \qw & \qw        & \qw        & \qw        & \qw        & \qw        & \qw        & \qw        & \qw        & \qw        & \qw        & \qw        & \gate{U_4} & \qw        & \qw        & \qw        & \qw        & \qw        & \qw        & \qw        & \qw        & \qw        & \gate{U_{0}} & \qw        & \qw        & \qw        & \qw        & \qw        & \qw        & \qw        & \qw        & \qw        & \gate{U_4} & \qw        & \qw        & \qw        & \qw        & \qw        & \qw        & \qw        & \qw        & \qw        & \qw        & \qw        & \qw &                                      \\
\lstick{\ket{\Psi_{8,5}}} & \qw & \qw        & \qw        & \qw        & \qw        & \qw        & \qw        & \qw        & \qw        & \gate{U_5} & \qw        & \qw        & \qw        & \qw        & \qw        & \qw        & \qw        & \qw        & \qw        & \qw        & \qw        & \qw        & \gate{U_{0}} & \qw        & \qw        & \qw        & \qw        & \qw        & \qw        & \qw        & \qw        & \qw        & \qw        & \qw        & \qw        & \gate{U_5} & \qw        & \qw        & \qw        & \qw        & \qw        & \qw        & \qw        & \qw        & \qw &                                      \\
\lstick{\ket{\Psi_{8,6}}} & \qw & \qw        & \qw        & \qw        & \qw        & \qw        & \gate{U_6} & \qw        & \qw        & \qw        & \qw        & \qw        & \qw        & \qw        & \qw        & \qw        & \qw        & \qw        & \qw        & \qw        & \qw        & \qw        & \gate{U_{0}} & \qw        & \qw        & \qw        & \qw        & \qw        & \qw        & \qw        & \qw        & \qw        & \qw        & \qw        & \qw        & \qw        & \qw        & \qw        & \gate{U_6} & \qw        & \qw        & \qw        & \qw        & \qw        & \qw &                                      \\
\lstick{\ket{\Psi_{8,7}}} & \qw & \qw        & \qw        & \gate{U_7} & \qw        & \qw        & \qw        & \qw        & \qw        & \qw        & \qw        & \qw        & \qw        & \qw        & \qw        & \qw        & \qw        & \qw        & \qw        & \qw        & \qw        & \qw        & \gate{U_{0}} & \qw        & \qw        & \qw        & \qw        & \qw        & \qw        & \qw        & \qw        & \qw        & \qw        & \qw        & \qw        & \qw        & \qw        & \qw        & \qw        & \qw        & \qw        & \gate{U_7} & \qw        & \qw        & \qw &                                      \\
\lstick{\ket{\Psi_{8,8}}} & \qw & \qw        & \qw        & \qw        & \qw        & \qw        & \qw        & \qw        & \qw        & \qw        & \qw        & \qw        & \qw        & \qw        & \qw        & \qw        & \qw        & \qw        & \qw        & \qw        & \qw        & \qw        & \gate{U_{0}} & \qw        & \qw        & \qw        & \qw        & \qw        & \qw        & \qw        & \qw        & \qw        & \qw        & \qw        & \qw        & \qw        & \qw        & \qw        & \qw        & \qw        & \qw        & \qw        & \qw        & \qw        & \qw &    
}$$}
\normalsize

\end{widetext}

\end{document}